\documentclass{llncs}

\usepackage[small, 
           compact
           ]{titlesecMod} 
\titlespacing{\section}{0pt}{*2}{*1} 
   
\usepackage{times} 

\titleformat{\paragraph}[runin]
   {\normalfont\normalsize\bfseries}{\theparagraph}{1em}{}

\AtBeginDocument{%
\addtolength\abovedisplayskip{-0.4\baselineskip}%
\addtolength\belowdisplayskip{-0.4\baselineskip}%
\addtolength{\textfloatsep}{-4ex}
\addtolength{\intextsep}{-4ex}
\addtolength{\abovecaptionskip}{-2ex}
\addtolength{\belowcaptionskip}{2ex}
}

\renewcommand\floatpagefraction{.9}
\renewcommand\topfraction{.9}

\usepackage{amssymb}
\usepackage{amsmath}
\usepackage{stmaryrd}
\usepackage{enumerate}
\usepackage{semantic}
\usepackage{changepage}
\usepackage{color}
\usepackage{mathpartir}
\usepackage{listings}
\usepackage{xifthen}
\usepackage{tikz}
\usetikzlibrary{arrows,automata}
\usetikzlibrary{positioning}

\mathlig{<}{\langle}
\mathlig{>}{\rangle}
\mathlig{==>}{\Rightarrow}
\mathlig{|-}{\mathop{\vdash}}
\renewcommand{\and}{\mathop{\wedge}}

\newcommand{\length}[1]{|#1|}

\makeatletter
\renewcommand{\p@enumii}{\theenumi.}
\makeatother

\usepackage{ifthen}
\usepackage{xargs}

\newcommandx{\yaHelper}[2][1=\empty]{%
\ifthenelse{\equal{#1}{\empty}}%
  { \ensuremath{ \scriptstyle{ #2 } } } 
  { \raisebox{ #1 }[0pt][0pt]{ \ensuremath{ \scriptstyle{ #2 } } } }  
}   

\newcommandx{\yrightarrow}[4][1=\empty, 2=\empty, 4=\empty, usedefault=@]{%
  \ifthenelse{\equal{#2}{\empty}}
  { \xrightarrow{ \protect{ \yaHelper[ #4 ]{ #3 } } } } 
  { \xrightarrow[ \protect{ \yaHelper[ #2 ]{ #1 } } ]{ \protect{ \yaHelper[ #4 ]{ #3 } } } } 
}

\newcommand{\overarrow}[2]{\textnormal{\raisebox{-0.5ex}{$\yrightarrow{#1}[-1pt]\mathrel{\vphantom{\to}^{#2}}$}}}
\newcommand{\overarrowA}[2]{\textnormal{\raisebox{-0.2ex}{$\yrightarrow{#1}[-1pt]\raisebox{-0.3ex}{$\hspace{-0.5ex}\mathrel{\overset{#2}{\mbox{\scriptsize $\!a$}}}$}$}}}

\newcommand{\Gammaid}{{\Gamma_{\mathit{id}}}}

\newcommand{\ensuretext}[1]{\ensuremath{\text{#1}}}
\newcommand{\ruleName}[1]{\ensuretext{\small \textsc{#1}}}
\newcommand{\TSSkip}{TS-Skip}
\newcommand{\TSAssign}{TS-Assign}
\newcommand{\TSOutput}{TS-Output}
\newcommand{\TSSequence}{TS-Seq}
\newcommand{\TSIfElse}{TS-IfElse}
\newcommand{\TSWhile}{TS-While}

\newcommand{\pc}{{\tt pc}}

\newcommand{\PVar}{{\it PVar}}
\newcommand{\Var}{{\it Var}}
\newcommand{\Chan}{{\it Chan}}
\newcommand{\chan}{\textsf{\small channel}}
\newcommand{\val}{\textsf{\small value}}
\newcommand{\PPoint}{{\it PPoint}}
\newcommand{\CPoint}{{\it CPoint}}

\newcommand{\Gequiv}[1]{\mathop{=_{\Gamma(#1)}}}
\newcommand{\NGequiv}[2]{\mathop{=_{#1(#2)}}}

\newcommand{\draftnote}[1]{}

\lstdefinelanguage{pollang}{  
  language     = java,
  basicstyle   = \ttfamily \footnotesize,
  keywordstyle = \bfseries \ttfamily \footnotesize, 
  stringstyle  = \color{javared},
  stepnumber   = 1, 
  numberstyle  = \color[rgb]{0.5,0.5,0.5}\ttfamily\scriptsize,
  tabsize      = 4,
  captionpos   = b,
  breaklines   = true,
  breakautoindent = false,
  escapeinside = {/*@}{@*/},
  aboveskip    = 8pt,
  belowcaptionskip = 0pt,
  morecomment  = [l]{//},
  mathescape   = true,
  showstringspaces = false
}
\newboolean{isTR}
\newcommand{\paperTR}[2]{\ifthenelse{\boolean{isTR}}{#2}{#1}}

\newcommand{\compl}[1]{\overline{#1}}

\newcommand{\ek}{\mathit{ek}} %

\newcommand{\SetDef}[2]{\{ #1~|~#2 \}}
\newcommand{\eclass}[3][\sigma']{\SetDef{#1}{#1 #3 #2}}
\newcommand{\compeclass}[3][\rho]{\eclass[#1]{#2}{\not#3}} 

\usepackage{wrapfig}
\usepackage{float}
\usepackage{subcaption}
\setboolean{isTR}{true}

\begin{document}
\mainmatter





\mainmatter
\title{Very Static Enforcement of Dynamic Policies}
\author{ Bart van Delft$^1$
    \and Sebastian Hunt$^2$
    \and David Sands$^1$ }
\institute{Chalmers University of Technology, Sweden
        \and City University London}
\maketitle

\begin{abstract}

Security policies are naturally dynamic. Reflecting this, there has been a growing interest in studying information-flow properties which change during program execution, including concepts such as declassification, revocation, and role-change.  

A static verification of a dynamic information flow policy, from a semantic perspective, should only need to concern itself with two things: 1) the dependencies between data in a program, and 2) whether those dependencies are consistent with the intended flow policies as they change over time. In this paper we provide a formal ground for this intuition. We present a straightforward extension to the principal flow-sensitive type system introduced by Hunt and Sands (POPL '06, ESOP '11) to infer both end-to-end dependencies and dependencies at intermediate points in a program. This allows typings to be applied to verification of both static and dynamic policies. Our extension preserves the principal type system's distinguishing feature, that type inference is independent of the policy to be enforced: a single, generic dependency analysis (typing) can be used to verify many different dynamic policies of a given program, thus achieving a clean separation between (1) and (2).

We also make contributions to the foundations of dynamic information flow. Arguably, the most compelling semantic definitions for dynamic security conditions in the literature are phrased in the so-called knowledge-based style. We contribute a new definition of knowledge-based termination insensitive security for dynamic policies.  We show that the new definition avoids anomalies of previous definitions and enjoys a simple and useful characterisation as a two-run style property.

\end{abstract}

\section{Introduction}
\label{sec:dependency}
\newcommand\blfootnote[1]{%
  \begingroup
  \renewcommand\thefootnote{}\footnote{#1}%
  \addtocounter{footnote}{-1}%
  \endgroup
}

\paperTR{}{\blfootnote{This is the Technical Report of~\cite{delft2015}.}}

Information flow policies are security policies
which aim to provide end-to-end security guarantees of the
form ``information must not flow from this source to this destination''.
Early work on information flow concentrated on static, multi-level policies,
organising the various data sources and sinks of a system into a fixed hierarchy.
The policy determined by such a hierarchy (a partial order) is simply that information must not flow
from $a$ to $b$ unless $a \sqsubseteq b$.

\subsection{Dynamic policies}
\label{subsec:intro:dynpol}
Since the pioneering work of Denning and Denning \cite{Denning:Denning:Certification}, a wide
variety of infor{-}mation-flow policies and corresponding enforcment mechanisms have been proposed.
Much recent work on information-flow properties goes beyond the static, multi-level security policies
of earlier work, considering instead more sophisticated, dynamic forms of policy which permit different flows
at different points during the excecution of a program.
Indeed, this shift of focus better reflects real-world requirements for security policies which are naturally dynamic.

\begin{wrapfigure}{r}{0.22\textwidth}
  \vspace{-1em}
  \begin{lstlisting}
  // /*@$ {\tt x} \rightarrow a $@*/;
  out x on /*@$a$@*/;
  // /*@$ {\tt x} \not\rightarrow a $@*/;
  out 2 on /*@$a$@*/;
  \end{lstlisting}
  \caption{}
  \vspace{-2.5em}
  \label{fig:depex1}
\end{wrapfigure}

For example, consider a request for sensitive employee information made to an employer
by a regulatory authority.
In order to satisfy this request it may be necessary to
temporarily allow the sensitive information to flow to a specific user in the Human Resources department.
In simplified form, the essence of this example is captured in Figure~\ref{fig:depex1}.

Here ${\tt x}$ contains the sensitive information, channel $a$ represents the HR user,
and the policy is expressed by the annotations
$ {\tt x} \rightarrow a $ (${\tt x}$ \emph{may} flow to $a$)
and
$ {\tt x} \not\rightarrow a $ (${\tt x}$ \emph{must not} flow to $a$).
It is intuitively clear that this program  complies with the policy.

Consider two slightly more subtle examples, in each of which revocation of a permitted flow
depends on run-time data:

\begin{wrapfigure}{l}{0.55\textwidth}
  \vspace{-1em}
  \begin{lstlisting}[numbers=left]
/*Program A*/       /*Program B*/
// /*@$ {\tt x, y} \rightarrow a $@*/;          // /*@$ {\tt x} \rightarrow a $@*/;
out x on /*@$a$@*/;         out x on /*@$a$@*/;
if (y > 0) {        if (x > 0) {
  out 1 on /*@$a$@*/;         out 1 on /*@$a$@*/;
// /*@$ {\tt x} \not\rightarrow a $@*/; /*@$\ \!$@*/          // /*@$ {\tt x} \not\rightarrow a $@*/;
}                   }
out 2 on /*@$a$@*/;         out 2 on /*@$a$@*/;    
out 3 on /*@$a$@*/;         out 3 on /*@$a$@*/;
  \end{lstlisting}
\end{wrapfigure}
In program A, the revocation of
${\tt x} \rightarrow a $
is controlled by the value of ${\tt y}$, whereas in program B
it is controlled by the value of ${\tt x}$ itself.
Note that the policy for A explicitly allows 
${\tt y} \rightarrow a $
so the conditional output (which reveals information about ${\tt y}$) appears to
be permissible.
In program B the conditional output reveals information about ${\tt x}$ itself,
but this happens \emph{before} the revocation. So should program B be regarded as compliant?
We argue that it should not, as follows.
Consider ``the third output'' of program B as observed on channel $a$.
Depending on the initial value of \texttt{x}, the observed value may be
either 2 (line 8) or 3 (line 9).
Thus this observation reveals information about ${\tt x}$ and,
in the cases where revocation occurs, the observation happens \emph{after}
the revocation.

Unsurprisingly,
increasing the sophistication of policies also increases the challenge of formulating good semantic definitions, which is to say,
definitions which both match our intuitions about what the policies mean and can form the basis of formal
reasoning about correctness.

At first sight it might seem that increasing semantic sophistication
should also require increasingly intricate enforcement mechanisms.
However, all such mechanisms must somehow solve the same two distinct problems:
\begin{enumerate}

   \item
   Determine what data dependencies exist between the various data sources and sinks manipulated by the program.

   \item
   Determine whether those dependencies are consistent with the flows permitted by the policy.

\end{enumerate}
Ideally, the first of these problems would be solved independently of the second, since
dependencies are a property of the code, not the policy. This would allow reuse at two levels:
a) reuse of the same dependency analysis mechanisms and proof techniques for different \emph{types} of policy;
b) reuse of the dependency properties for a given program across verification of multiple \emph{alternative} policies
(whether of the same type or not).

In practice, enforcement mechanisms are typically not presented in a way which cleanly separates the two concerns.
Not only does this hamper the reuse of analysis mechanisms and proof techniques,
it also makes it harder to identify the \emph{essential} differences between
different approaches.

\paragraph{Central Contribution} 
We take a well-understood dependency type system for a simple while-language,
originally designed to support enforcement of static policies,
and extend it in a straightforward way to a language with output
channels (\S~\ref{sec:typesystem}).
We demonstrate the advantages of a clean separation between dependency analysis and policy enforcement,
by establishing a generic soundness result (\S~\ref{sec:semsoundness}) for the type system
which characterises the meaning of types as dependency properties.
We then show how the dependency information derived by the type system
can be used to verify compliance with dynamic policies.
Note that this means that the core analysis for enforcement can be done even before the policy is
known: we dub this \emph{very static} enforcement.
More significantly, it opens the way to reuse of dependency analyses across verification of multiple types
of information flow policy (for example, it might be possible
to use the dependency analyses performed by advanced slicing
tools such as Joanna and Indus).

\paragraph{Foundations of Dynamic Flow Policies}
Although it was not our original aim and focus, we also make some
contributions of a more foundational nature, and our paper opens
with these
(\S\ref{sec:policymodel}--\S\ref{sec:secproperty}). The 
semantic definition of security which we use is based on
work of Askarov and
Chong \cite{askarov2012}, and we begin with their abstract formulation
of dynamic policies (\S\ref{sec:policymodel}).
In defining security for dynamic policies, they made a convincing case for using a
family of attackers of various strengths, following an observation
that the intuitively strongest attacker (who never forgets anything
that has been observed) actually places weaker security demands on the
system than we would want. On the other hand they observe that the family of
\emph{all} attackers contains pathological attacker behaviours which one
certainly does not wish to consider. Due to this they do not give a
 characterisation of the set of all \emph{reasonable} attackers
against which one should protect. We make the following two
foundational contributions:

\paragraph{Foundational Contribution 1} We focus (\S\ref{subsec:pi}) on the pragmatic case
of \emph{progress insensitive} security (where slow information leakage is allowed
through observation of computational progress
\cite{Askarov+:ESORICS08}). We argue for a new definition of progress
insensitive security (Def \ref{def:pikbsec}), which unconditionally grants all attackers
knowledge of computational progress. With this modification to 
the definition from \cite{askarov2012}, 
the problematic examples of pathological attackers are eliminated, and
we have a more complete definition of security. Consequently, we are able
to prove security in the central contribution of the paper \emph{for all attackers}.

\paragraph{Foundational Contribution 2}
The definitions of security are based on characterising attacker
knowledge and how it changes over time relative to the changing
policy. 
As argued previously e.g., \cite{Broberg:Sands:PLAS09}, this style of
definition forms a much more intuitive basis for a semantics of
dynamic policies than using two-run characterisations. However,
two-run formulations have the advantage of being easier to use in
proofs. We show (\S\ref{sec:secproperty}) 
that our new knowledge-based progress-insensitive
security definition enjoys a simple two-run characterisation. We make
good use of this in our proof of correctness of our central contribution.

\section{The Dynamic Policy Model}
\label{sec:policymodel}
In this section we define an abstract model of computation and a model of
dynamic policies which maps computation histories to equivalence
relations on stores.

\subsection{Computation and Observation Model}
\label{sec:computation}

\paragraph{Computation Model}
The computation model is given by a labelled transition system over
\emph{configurations}. We write 
$<c,\sigma> \overarrow{\alpha}{} <c',\sigma'>$ means that
configuration $<c,\sigma>$ evaluates in one step to configuration
$<c', \sigma'>$ with label $\alpha$. Here $c$ is a \emph{command} and
$\sigma \in \Sigma$ is a \emph{store}. In examples and when we 
instantiate this model the store will be a mapping from program
variables to values.

The label $\alpha$ records any output that happens during that
step, and we have a distinguished label value $\epsilon$ to denote a
silent step which produces no output. Every non-silent label $\alpha$
has an associated
channel $\chan(\alpha) \in \Chan$ and a value $\val(\alpha)$. Channels
are ranged over by $a$ and values by $v$.
We abbreviate a sequence of evaluation steps 
\[
<c_0,\sigma_0> \overarrow{\alpha_1}{} <c_1,\sigma_1> \overarrow{\alpha_2}{} \dots \overarrow{\alpha_n}{}
<c_n,\sigma_n> 
\]
 as $<c_0,\sigma_0> \overarrow{}{n} <c_n,\sigma_n>$.
We write $<c_0,\sigma_0> \overarrow{}{*} <c',\sigma'>$ if $<c_0,\sigma_0> \overarrow{}{n} <c',\sigma'>$
for some $n \mathop{\geq} 0$.
We write the projection of a single step $<c,\sigma> \overarrow{\alpha}{} <c',\sigma'>$
to some channel~$a$ as $<c,\sigma> \overarrowA{\beta}{} <c',\sigma'>$ 
where $\beta \mathop{=} v$ if $\chan(\alpha) = a$ and $\val(\alpha) =
v$, 
 and $\beta \mathop{=} \epsilon$ otherwise,
that is, when $\alpha$ is silent or an output on a channel different from $a$.

We abbreviate a sequence of evaluation steps 
\[ <c_0,\sigma_0> \overarrowA{\beta_1}{} <c_1,\sigma_1> \overarrowA{\beta_2}{} \dots \overarrowA{\beta_n}{}
<c_n,\sigma_n> 
\]
as $<c_0,\sigma_0> \overarrowA{t}{n} <c_n,\sigma_n>$
where $t$ is the trace of values produced on channel $a$ with every silent $\epsilon$ filtered out.
We write $<c_0,\sigma_0> \overarrowA{t}{*} <c',\sigma'>$ if $<c_0,\sigma_0> \overarrowA{t}{n} <c',\sigma'>$
for some $n \mathop{\geq} 0$.

We use $\length{t}$ to denote the length of trace $t$ and
$t_1 \preceq t_2$ to denote that trace~$t_1$ is a prefix of trace~$t_2$.

\paragraph{Attacker's Observation Model}
We follow the standard assumption that the command $c$ is known to the attacker.
We assume a passive attacker which aims to extract information about an
input store $\sigma$ by observing outputs.
As in \cite{askarov2012}, the attacker is able only to
observe a \emph{single} channel.
A generalisation to multi-channel attackers (which would also allow
colluding attackers to be modelled) is left for future work.

\subsection{Dynamic Policies}
\label{sec:dynpolicies}
A flow policy specifies a limit on how much information an attacker may learn.
A very general way to specify such a limit is as an equivalence relation on input stores.

\begin{example} Consider a store with variables {\tt x} and {\tt y}.
A simple policy might state that the attacker should only be able to
learn the value of {\tt x}. It follows that all stores which agree on the value of
{\tt x} should look the same to the attacker. This is expressed as
the equivalence relation $\sigma \mathop{\equiv} \rho$ iff $\sigma({\tt x}) \mathop{=} \rho({\tt x})$.

A more complicated policy might allow the attacker to learn the value of
some arbitrary expression $e$ on the initial store, e.g.\ ${\tt x} \mathop{=} {\tt y}$.
This is expressed as the equivalence relation 
 $\sigma \mathop{\equiv} \rho$ iff $\sigma(e) \mathop{=} \rho(e)$.
\end{example}

\begin{definition} [Policy]
  A policy $P$ maps each channel to an equivalence relation $\equiv$ on
  stores. We write $P_a$ for the equivalence relation that $P$ defines for channel $a$.
\end{definition}

As defined, policies are static.
A dynamic policy changes while the program is running and
may dictate a different $P$ for each point in the execution.
Here we assume that the policy changes \emph{synchronously} with the execution of
the program. That is, the active policy can be deterministically derived
from the execution history so far.

\begin{definition} [Execution History]
  An execution history $\mathcal{H}$ of length $n$ is a transition sequence
  $<c_0,\sigma_0> \overarrow{\alpha_1}{} <c_1,\sigma_1> \overarrow{\alpha_2}{} \dots \overarrow{\alpha_n}{}
<c_n,\sigma_n>$.
\end{definition}

\begin{definition} [Dynamic Policy]
\label{def:dynamicpol}
  A dynamic policy $D$ maps every execution history $\mathcal{H}$ to a policy $D(\mathcal{H})$.
  We write $D_a(\mathcal{H})$ for the equivalence relation that is defined by
  $D(\mathcal{H})$ for channel $a$, that is to say, $D_a(\mathcal{H}) = P_a$ where
  $P = D(\mathcal{H})$.
\end{definition}

Most synchronous dynamic policy languages in the literature determine the
current policy based solely on the store $\sigma_n$ in the final
configuration of the execution history \cite{askarov2012,Paragon}.
Definition~\ref{def:dynamicpol} allows in principle for more flexible notions of dynamic policies,
as they can incorporate the full execution history to determine
the policy at each stage of an execution
(similar to the notion of conditional noninterference used by \cite{Goguen:Meseguer:Unwinding,Zhang2012}).
However, our enforcement does assume that the
dynamic policy can be statically approximated per program point, which arguably
is only feasible for policies in the style of \cite{askarov2012,Paragon}.
Such approximations can typically be improved by allowing the program to branch on policy-related queries.

Since programs are deterministic, an execution history of length $n$ is uniquely determined
by its initial configuration $<c_0,\sigma_0>$. We use this fact to simplify our definitions
and proofs:
\begin{definition} [Execution Point]
  An execution point is a triple $(c_0,\sigma_0,n)$ identifying the point in
  execution reached after $n$ evaluation steps starting from configuration  
  $<c_0,\sigma_0>$. Such an execution point is considered well-defined iff
  there exists $<c_n,\sigma_n>$ such that $<c_0,\sigma_0> \overarrow{}{n} <c_n,\sigma_n>$.
\end{definition}
\begin{lemma}
Each well-defined execution point $(c_0,\sigma_0,n)$ uniquely determines
an execution history $\mathcal{H}(c_0,\sigma_0,n)$ of length $n$ starting
in configuration $<c_0,\sigma_0>$.
\end{lemma}
In the rest of the paper we rely on this fact to justify a convenient abuse
of notation, writing $D(c_0,\sigma_0,n)$ to mean $D(\mathcal{H}(c_0,\sigma_0,n))$.

\section{Knowledge-Based Security Conditions}
\label{sec:knowledgebased}
Recent works on dynamic policies,
including \cite{askarov2012,Balliu:2011:ETL,Banerjee+:Expressive,Broberg:Sands:Paralocks},
make use of so-called \emph{knowledge-based}
security definitions, building on the notion of gradual release introduced 
in~\cite{Askarov:Sabelfeld:Gradual}.
This form of definition seems well-suited to provide intuitive semantics
for dynamic policies.
We focus in particular on the attacker-parametric model from 
Askarov and Chong in \cite{askarov2012}.

Suppose that the input state to a program is $\sigma$.
In the knowledge-based approach, an attacker's knowledge of $\sigma$
is modelled as a \emph{knowledge set} $K$, which may be any
set of states such that $\sigma \in K$.
Note that the larger the knowledge set, the less certain is the attacker
of the actual value of $\sigma$, so smaller $K$ means more precise knowledge.
(Sometimes, as we see below, it can be more intuitive to focus on
the complement $\compl{K}$, which is the set of a-priori possible states which the attacker is
able to \emph{exclude}, since this set, which we will call the
\emph{exclusion knowledge},
grows as the attacker learns more).

Now suppose that the currently active policy is $\equiv$. The essential
idea in any know{-}ledge-based semantics is to view the equivalence classes of
$\equiv$ as placing upper bounds on the attacker's knowledge. In the simplest
setting, if the actual input state is $\sigma$ and the attacker's knowledge
set is $K$, we require:
\[
	K \supseteq \eclass{\sigma}{\equiv} 
\]
Or, in terms of what the attacker is able to exclude:
\begin{equation}
\label{eqn:K-bound}
	\compl{K} \subseteq \compeclass{\sigma}{\equiv} 
\end{equation}

How then do we determine the attacker's knowledge?
Suppose an attacker knows the program $c$ and observes channel $a$.
Ignoring covert channels (timing, power, etc) an obvious approach is to
say that the attacker's knowledge is simply a function of the trace $t$ observed so far:
\begin{equation}
\label{eqn:K-proto}
	k(t) = \{ \rho | <c,\rho> \overarrowA{t}{} \}
\end{equation}
We define the exclusion knowledge as the complement of this: $\ek(t) = \compl{k(t)}$.
Note that, as a program executes and more outputs are observed, the
attacker's exclusion knowledge
can only increase; if $<c,\sigma> \overarrowA{t \cdot v}{}$ then
\(
	{\ek(t)} \subseteq {\ek(t\cdot v)},
\)
since, if $\rho$ can already be excluded by observation of $t$ (because $c$ cannot produce
$t$ when started in $\rho$), then it will still be excluded when $t\cdot v$ is observed
(if $c$ cannot produce $t$ it cannot produce any extension of $t$ either).

But this simple model is problematic for dynamic policies.
Suppose that the policies
in effect when $t$ and $t\cdot v$ are observed are, respectively
$\equiv_{1}$ and $\equiv_{2}$.
Then it seems that we must require both
${\ek(t)} \subseteq \compeclass{\sigma}{\equiv_{1}}$ 
and
${\ek(t\cdot v)} \subseteq \compeclass{\sigma}{\equiv_{2}}$. 
As observed above, it will always be the case that
${\ek(t)} \subseteq {\ek(t\cdot v)}$, so we are forced to require
${\ek(t)} \subseteq \compeclass{\sigma}{\equiv_{2}}$.
In other words, the observations that we can permit at any given
moment are constrained not only by the policy currently in effect but also
by all policies which will be in effect in the future. This makes it impossible to have dynamic
policies which revoke previously permitted flows (or, at least, pointless; since all
revocations would apply retrospectively, the earlier ``permissions'' could never be exercised).

Askarov and Chong's solution has two key components, outlined in the following two sections.
\subsection{Change in Knowledge}
Firstly, recognising that policy changes should not apply retrospecively,
we can relax (\ref{eqn:K-bound}) to constrain only how an attacker's knowledge should be allowed
to \emph{increase}, rather than its absolute value. The increase in attacker knowledge
going from $t$ to $t\cdot v$ is
given by the set difference ${\ek(t\cdot v)} - {\ek(t)}$. 
So, instead of (\ref{eqn:K-bound}), we require:
\begin{equation}
\label{eqn:K-increase-bound}
	{\ek(t\cdot v)} - {\ek(t)} \subseteq \compeclass{\sigma}{\equiv}
\end{equation}
where $\equiv$ is the policy in effect immediately before the output $v$.
(Some minor set-theoretic rearrangement gives the equivalent
\[
	k(t\cdot v) \supseteq k(t) \cap \eclass{\sigma}{\equiv}
\]
which is the form of the original presentation in \cite{askarov2012}.)

\subsection{Forgetful attackers}
\label{sec:knowledgebased:forgetful}
Focussing on change in knowledge addresses the problem of retrospective revocation but it creates a new issue.
Consider the following example.

\begin{wrapfigure}{r}{0.28\textwidth}
  \begin{lstlisting}
/*@$ {\tt x} \rightarrow a $@*/;
out x on a;
/*@$ {\tt x} \not\rightarrow a $@*/;
while (true)
  out x on a;
  \end{lstlisting}
  \caption{}
  \vspace{-1em}
  \label{fig:kbex1}
\end{wrapfigure}
  
\begin{example}
  \label{ex:gradualrelease}
  The program in Figure~\ref{fig:kbex1} produces the same output many times,
  but only the first output is permitted by the policy.
  Assume that the value of {\tt x} is $5$.
  Before the first output, the knowledge set of an observer on channel $a$ contains
  every possible store.
  After the first output the observer's knowledge set is reduced to include only
  those stores $\sigma$ where $\sigma({\tt x}) = 5$. This
  is allowed by the policy at that point.

  By the time the second output occurs, the policy prohibits any further
  flow from {\tt x}. However, since the attacker's knowledge set \emph{already}
  includes complete knowledge of {\tt x}, the second output does not
  actually change the attacker's knowledge at all, so
  (\ref{eqn:K-increase-bound}) is satisfied (since $k(t\cdot v) = k(t)$).
  Thus a policy semantics based on (\ref{eqn:K-increase-bound}) would accept this program
  even though it continues to leak the value of {\tt x} \emph{long} after the
  flow has been revoked.

\end{example}

Askarov and Chong address this by revisiting the assumption that
an attacker's knowledge is necessarily determined by the simple function of traces
(\ref{eqn:K-proto}) above.
Consider an attacker which \emph{forgets} the value of the first output in example~\ref{ex:gradualrelease}.
For this attacker, the second output would come as a revalation, revealing the value
of {\tt x} all over again, in violation of the policy.
Askarov and Chong thus arrive at the intriguing observation that security against a more
powerful attacker, one who remembers everything that happens,
does not imply security against a less resourceful attacker,
who might forget parts of the observations made.

Forgetful attackers are modelled as deterministic automata.

\begin{definition}[Forgetful Attacker $\rhd$ \S\ III.A~\cite{askarov2012}]
A forgetful attacker is a tuple $A\mathop{=}(S_A,s_0,\delta_A)$ where $S_A$ is the set of
attacker states; $s_0 \in S_A$ is the initial state; and
$\delta_A : S_A \times {\it Val} \rightarrow S_A$ the (deterministic) transition function
describing how the attacker's state changes due to the values that the
attacker observes.
\end{definition}

We write $A(t)$ for the attacker's state after observing trace $t$:
\begin{align*}
A(\epsilon) =&\ s_0 \\
A(t \cdot v) =&\ \delta_A(A(t),v)
\end{align*}
A forgetful attacker's knowledge after trace $t$ is defined as the set of all initial stores
that produce a trace which would result in the same attacker state $A(t)$:

\begin{definition}[Forgetful Attacker Knowledge $\rhd$ \S\ III.A~\cite{askarov2012}]
$$k(A,c,a,t) =
  \{ \rho \mid <c,\rho> \overarrowA{t'}{} \wedge A(t') = A(t)\}
$$
\end{definition}
(Note that, in preparation for the formal definition of the security condition,
program $c$ and channel $a$ now appear as explicit parameters.)


The proposed security condition is still essentially as given by (\ref{eqn:K-increase-bound}),
but now relative to a specific choice of attacker. Stated in the notation and style of the current paper,
the formal definition is as follows.

\begin{definition}[Knowledge-Based Security $\rhd$ Def. 1~\cite{askarov2012}]
\label{def:pskbsec}
Command $c$ is secure for policy $D$ against an attacker $A$ on channel $a$ for initial
store $\sigma$ if for all traces $t$ and values $v$ such
that $<c,\sigma>\overarrowA{t}{n}<c',\sigma'>\overarrowA{v}{1}$ 
we have
$$
{\ek(A,c,a,t\cdot v)} - {\ek(A,c,a,t)} \subseteq \compeclass{\sigma}{\equiv}
$$
where ${\equiv} = D_a(c,\sigma,n)$.
\end{definition}

Having relativized security to the power of an attacker's memory, it is natural
to consider the strong notion of security that would be obtained by requiring
Def.~\ref{def:pskbsec} to hold for all choices of $A$.
However, as shown in \cite{askarov2012}, this exposes a problem with the model:
there are attackers for which even well-behaved programs are insecure
according to Def.~\ref{def:pskbsec}.

\begin{example}
  Consider again the first example from the Introduction (Section~\ref{subsec:intro:dynpol}).
  Here, for simplicity, we assume that the variable {\tt x}
  is boolean, taking value 0 or 1.

  
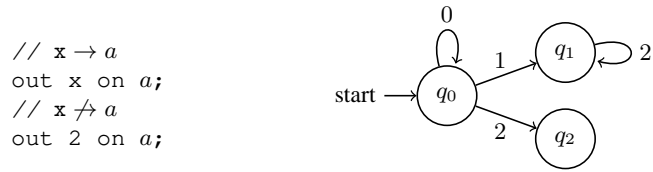
\begin{figure}[h]
  \centering
  \begin{subfigure}[b]{0.3\textwidth}
    \centering
  \begin{lstlisting}
  // /*@$ {\tt x} \rightarrow a $@*/
  out x on /*@$a$@*/;
  // /*@$ {\tt x} \not\rightarrow a $@*/
  out 2 on /*@$a$@*/;
  \end{lstlisting}
  \end{subfigure}
  \begin{subfigure}[b]{0.5\textwidth}
    \centering
    \begin{tikzpicture}[->,auto,node distance=2cm,semithick]
      \tikzstyle{every state}=[fill=none,draw=black,text=black]

      \node[initial,state] (A)                                {$q_0$};
      \node[state]         (B) [above right=0cm and 1cm of A] {$q_1$};
      \node[state]         (C) [below right=0cm and 1cm of A] {$q_2$};

      \path (A) edge [loop above]     node {$0$} (A)
                edge [above, pos=0.4] node {$1$} (B)
                edge [below, pos=0.4] node {$2$} (C)
            (B) edge [loop right]     node {$2$} (B);
    \end{tikzpicture}
  \end{subfigure}
\end{figure}
  
  It is intuitively clear that this program complies with the policy. However,
  as observed in \cite{askarov2012},
  if we instantiate Def.~\ref{def:pskbsec} with the forgetful attacker displayed,
  the attacker's knowledge increases with the second output
  when ${\tt x} \mathop{=} 0$.
  
  After observing the value $0$, the attacker's state is $A(0)\mathop{=}q_0$.
  Since $A(\epsilon) \mathop{=} q_0$, the knowledge set still holds every store
  possible. After the second observation, only stores where ${\tt x} \mathop{=} 0$
  could have led to state $q_2$, so the knowledge set shrinks (ie, the attacker's knowledge
  increases) at a point where the policy does not allow it.

\end{example}
This example poses a question which (so far as we are aware) remains unanswered:
if we base a dynamic policy semantics on Def.\ref{def:pskbsec}, for
\emph{which set} of attackers should we require it to hold?

In the next section we define a progress-insensitive variant of Def.\ref{def:pskbsec}.
For this variant it seems that security against all attackers \emph{is} a reasonable
requirement and in Section~\ref{sec:semsoundness} we show that progress-insensitive security
against all attackers is indeed enforced by our type system.

\subsection{Progress Insensitive Security}
\label{subsec:pi}
Since \cite{Volpano:Smith:Irvine:Sound},
 work on the formalisation and enforcement of information-flow
policies has generally distinguished between two flavours of security:
\emph{termination-sensitive} and \emph{termination-insensitive}.
Termination-sensitive properties guarantee that protected information is neither revealed
by its influence on input-output behaviour nor by its influence on termination behaviour.
Termination-insensitive properties allow the latter flows and thus provide weaker guarantees.
For systems with incremental
output (as opposed to batch-processing systems) it is more appropriate to
distinguish between 
\emph{progress-sensitive} and \emph{progress-insensitive} security.
Progress-insensitive security ignores progress-flows,
where a flow is regarded as a progress-flow if the information
that it reveals can be inferred solely by observing \emph{how many} outputs the system produces.

Two examples of programs with progress-flows are as follows:

\begin{example}
\label{ex:progressleaks}
  Programs containing progress-flows:
\begin{lstlisting}
// Program A            // Program B
out 1 on a;             out 1 on a;
while (x == 8) skip;    if (x != 8) out 2 on a;
out 2 on a;
\end{lstlisting}
Let $\sigma$ and $\rho$ differ only on the value of \texttt{x}:
$\sigma({\tt x}) = 4$ and $\rho({\tt x}) = 8$.
Note that, if started in $\sigma$, both programs produce a trace of length
2 (namely, the trace $1\mathop{\cdot}2$)
whereas, if started in $\rho$, the maximum trace length is 1.
Thus, for both programs, observing just the length of the trace
produced can reveal information about \texttt{x}.
Note that, since termination
is not an observable event in the semantics, A and B are actually observably equivalent;
we give the two variants to emphasise that progress-flows may occur even in the
absence of loops.

\end{example}

In practice, most enforcement mechanisms only enforce progress-insensitive security. This is
a pragmatic choice since (a) it is hard to enforce progress-sensitive security without being
overly restrictive (typically, all programs which loop on protected data will be rejected),
and (b) programs which leak solely via progress-flows, leak slowly \cite{Askarov+:ESORICS08}.

Recall that Knowledge-Based Security (Def.~\ref{def:pskbsec}) places a bound
on the increase in an attacker's knowledge which is allowed to arise from observation of the next output event.
Askarov and Chong show how this can be weakened in a natural way
to provide a progress-insensitive property, by artificially strengthening the supposed previous knowledge
to already include progress knowledge. Their definition of progress knowledge
is as follows:
\begin{definition}[AC Progress Knowledge $\rhd$ \S\ III.A~\cite{askarov2012}]
$$k^{+}(A,c,a,t) =
    \{ \rho \mid <c,\rho> \overarrowA{t'\cdot v}{} \wedge A(t') = A(t) \}
$$
\end{definition}
Substituting this (actually, its complement) in the ``previous knowledge'' position in Def.~\ref{def:pskbsec} provides
Askarov and Chong's notion of progress-insensitive security:
\begin{definition}[AC Progress-Insensitive (ACPI) Security $\rhd$ Def. 2~\cite{askarov2012}]
\label{def:akpisec}
Command $c$ is AC Progress-Insensitive secure for policy $D$ against an attacker $A$ on channel $a$ for initial
store $\sigma$ if for all traces $t$ and values $v$ such
that $<c,\sigma>\overarrowA{t}{n}<c',\sigma'>\overarrowA{v}{1}$
we have
$$
{\ek(A,c,a,t\cdot v)} - {\ek^{+}(A,c,a,t)} \subseteq \compeclass{\sigma}{\equiv}
$$
where ${\equiv} = D_a(c,\sigma,n)$.
\end{definition}

Now consider again programs A and B above. These are examples of programs where the \emph{only} flows
are progress-flows. In general, we say that a program is \emph{quasi-constant}
if there is some fixed (possibly infinite) trace $t$ such that every trace
produced by the program is a prefix of $t$, regardless
of the choice of initial store.
Thus, for a quasi-constant program, the only possible observable variation in observed
behaviour is trace length, so all flows are progress-flows.
Since PI security is intended explicitly to allow progress-flows,
we should expect all quasi-constant programs to satisfy PI security,
regardless of the choice of policy and for all possible attackers.
But, for Def.~\ref{def:akpisec}, this fails to hold, as shown by the following counterexample.

\begin{example}
\label{ex:piflawed}
  Consider the program and attacker below.
  The attacker is a very simple bounded-memory attacker which remembers just the last output seen
  and nothing else (not even whether it has seen any previous outputs).
  
  \begin{figure}[h]
  \centering
  \begin{subfigure}[b]{0.3\textwidth}
    \centering
  \begin{lstlisting}
// /*@${\tt x} \not\rightarrow a$@*/
out 1 on /*@$a$@*/;
out 1 on /*@$a$@*/;
while (x) skip;
out 1 on /*@$a$@*/;
out 2 on /*@$a$@*/;
  \end{lstlisting}
  \end{subfigure}
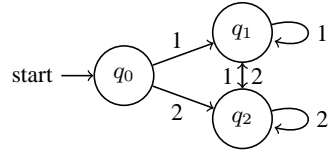
  \begin{subfigure}[b]{0.5\textwidth}
    \centering
    \begin{tikzpicture}[->,auto,node distance=2cm,semithick]
      \tikzstyle{every state}=[fill=none,draw=black,text=black]

      \node[initial,state] (A)                                {$q_0$};
      \node[state]         (B) [above right=0cm and 1cm of A] {$q_1$};
      \node[state]         (C) [below right=0cm and 1cm of A] {$q_2$};

      \path (A) edge [above, pos=0.4] node {1} (B)
                edge [below, pos=0.4] node {2} (C)
            (B) edge [loop right]     node {1} (B)
                edge                  node {2} (C)
            (C) edge [loop right]     node {2} (C)
                edge node {1} (B);
    \end{tikzpicture}
  \end{subfigure}
\end{figure}
  
  Clearly, the program is quasi-constant.
  However, it is \emph{not} ACPI secure for the given attacker.
  To see this, suppose that {\tt x} = 0 and consider the trace
  $t  = 1 \cdot 1 \cdot 1$.
  The attacker has no knowledge at this point ($k(t)$ is the set of all stores)
  since it does not know whether it has seen one, two or three 1's.
  It is easily verified that $k^{+}(t)$ is also the set of all stores for this attacker
  (intuitively, giving this attacker progress knowledge in the form $k^{+}$ doesn't help it,
  since it still does not know which side of the loop has been reached).
  But $k(t\cdot 2)$ is \emph{not} the set of all stores, since in state $q_2$ the
  attacker is able to exclude all stores for which {\tt x} = 1, thus ACPI security
  is violated.
\end{example}
What has gone wrong here? The attacker itself seems reasonable. We argue that the real problem lies in the definition of
$k^{+}(A,c,a,t)$. As defined, this is the knowledge that $A$ would have in state $A(t)$ if
given just the additional information that $c$ can produce at least one more output.
But this takes no account of any \emph{previous} progress knowledge which might
have been forgotten by $A$. (Indeed, the above attacker forgets nearly all such previous
progress knowledge.)
As a consequence, the resulting definition of PI security mistakenly treats some increases in knowledge
as significant, even though they arise purely because the attacker has forgotten previously
available progress knowledge.

Our solution will be to re-define progress knowledge to include what the attacker would
know \emph{if it had been keeping count}. To this end,
for any attacker $A = (S,s_0,\delta)$ we define a counting variant
$A^{\omega} = (S^{\omega},s^{\omega}_0,\delta^{\omega})$, such
that
$S^{\omega} \subseteq S \times N$,
$s^{\omega}_0 = (s_0, 0)$ and
$\delta^{\omega}((s,n),v) = (\delta(s,v), n+1)$.
In general, $A^\omega$ will be at least as strong an attacker as $A$:
\begin{lemma}
\label{lemma:omegastrengthens}
For all $A$, $c$, $a$, $t$:
\begin{enumerate}
\item $k(A^\omega,c,a,t) \subseteq k(A,c,a,t)$
\item $ek(A,c,a,t) \subseteq ek(A^\omega,c,a,t)$
\end{enumerate}
\end{lemma}
\begin{proof}
It is is easily seen that $A^\omega(t) = (q,n) \Rightarrow A(t) = q$.
Thus $A^\omega(t') = A^\omega(t) \Rightarrow A(t') = A(t)$,
which establishes part 1. Part 2 is just the contrapositive of part 1.
\end{proof}

Our alternative definition of progress knowledge is then:
\begin{definition}[Full Progress Knowledge]
$$k^{\#}(A,c,a,t) =
    \{ \rho \mid <c,\rho> \overarrowA{t'\cdot v}{} \wedge A^\omega(t') = A^\omega(t) \}
$$
\end{definition}
Our corresponding PI security property is:
\begin{definition}[Progress-Insensitive (PI) Security]
\label{def:pikbsec}
Command $c$ is progress-insensitive secure for policy $D$ against an attacker $A$ on channel $a$ for initial
store $\sigma$ if for all traces $t$ and values $v$ such
that $<c,\sigma>\overarrowA{t}{n}<c',\sigma'>\overarrowA{v}{1}$ 
we have
$$
{\ek(A,c,a,t\cdot v)} - {\ek^{\#}(A,c,a,t)}
	\subseteq \compeclass{\sigma}{\equiv}
$$
where ${\equiv} = D_a(c,\sigma,n)$.
\end{definition}
This definition behaves as expected for quasi-constant programs:
\begin{lemma}
\label{lemma:progress-only}
Let $c$ be a quasi-constant program.
Then $c$ is PI secure for all policies $D$ against all
attackers $A$ on all channels $a$ for all initial stores $\sigma$.
\end{lemma}
\begin{proof}
It suffices to note that, from the definitions,
if $t\cdot v$ is a possible trace for $c$ and $c$ is quasi-constant, then
$k^{\#}(A,c,a,t) = k(A^\omega,c,a,t\cdot v)$. The result  follows
by Lemma~\ref{lemma:omegastrengthens}.
\end{proof}

As a final remark in this section, we note that there is a class of attackers for which
ACPI and PI security coincide.
Say that $A$ is \emph{counting} if it
always remembers at least how many outputs it has observed. Formally:
\begin{definition}[Counting Attacker]
$A$ is counting if 
$A(t) = A(t') \Rightarrow \length{t} = \length{t'}$.
\end{definition}
Now say that attackers $A$ and $A'$ are isomorphic (written $A \cong A'$) if
$A(t_1) = A(t_2) \Leftrightarrow A'(t_1) = A'(t_2)$  and note that
none of the attacker-parametric security conditions distinguish
between isomorphic attackers
(in particular, knowledge sets are always equal for isomorphic attackers).
It is easily verified that $A \cong A^{\omega}$ for all counting attackers.
It is then immediate from
the definitions that ACPI security and PI security coincide for counting attackers.

\section{Progress-Insensitive Security as a Two-Run Property}
\label{sec:secproperty}
Our aim in this section is to derive a security property which guarantees (in fact, is equivalent to)
PI security for all attackers, and in a form which facilitates the soundness proof of our type system.
For this we seek a property in ``two run'' form.

First we reduce the problem by establishing that it suffices to consider just the counting attackers.
\begin{lemma}
\label{lemma:PI-counting}
Let $c$ be a command. Then, for any given policy, channel and initial store,
$c$ is PI secure against all attackers iff $c$ is PI secure against all counting attackers.
\end{lemma}
\begin{proof}
Left to right is immediate. Right to left, it suffices
to show that
\[ {\ek(A,c,a,t\cdot v)} - {\ek^{\#}(A,c,a,t)}
\subseteq
{\ek(A^\omega,c,a,t\cdot v)} - {\ek^{\#}(A^\omega,c,a,t)}
\]
Since $A^\omega \cong (A^\omega)^\omega$, we have ${\ek^{\#}(A^\omega,c,a,t)} = {\ek^{\#}(A,c,a,t)}$.
It remains to show that ${\ek(A,c,a,t\cdot v)} \subseteq {\ek(A^\omega,c,a,t\cdot v)}$,
which holds by Lemma~\ref{lemma:omegastrengthens}.
\end{proof}

Our approach is now essentially to unwind Def.~\ref{def:pikbsec}.
Our starting point for the unwinding is:
\[
	{\ek(A,c,a,t\cdot v)} - {\ek^{\#}(A,c,a,t)}
    \subseteq \compeclass{\sigma}{\equiv}
\]
where $\equiv$ is the policy in effect at the moment the output $v$ is produced.
Simple set-theoretic rearrangement gives the equivalent:
\[
	\eclass{\sigma}{\equiv} \cap k^{\#}(A,c,a,t) \subseteq k(A,c,a,t\cdot v)
\]
Expanding the definitions, we arrive at:
\[
	\rho \equiv \sigma
	\wedge
	<c,\rho> \overarrowA{t'\cdot v'}{}
	\wedge
	A^\omega(t') = A^\omega(t)
	\Rightarrow
	\exists s . <c, \rho> \overarrowA{s}{} \wedge A(s) = A(t\cdot v)
\]
By the above lemma, we can assume without loss of generality that $A$ is counting, so
we can replace $A^\omega(t') = A^\omega(t)$ by $A(t') = A(t)$ on the lhs.
Since $A$ is counting, we know that $\length{t} = \length{t'}$ and $\length{s} = \length{t\cdot v}$,
hence $\length{s} = \length{t'\cdot v'}$. Now, since $c$ is deterministic and both $s$ and $t'\cdot v'$
start from the same $\rho$, it follows that $s = t'\cdot v'$. Thus we can simplify the unwinding to:
\[
	\rho \equiv \sigma
	\wedge
	<c,\rho> \overarrowA{t'\cdot v'}{}
	\wedge
	A(t') = A(t)
	\Rightarrow
    A(t'\cdot v') = A(t\cdot v)
\]
Now, suppose that this holds for $A$ and that $v' \neq v$.
Let $q$ be the attacker state $A(t') = A(t)$
and let $r$ be the attacker state $A(t'\cdot v') = A(t\cdot v)$.
Since $\length{t} \neq \length{t\cdot v}$ and $A$ is counting,
we know that $q \neq r$.
Then we can construct an attacker $A'$ from $A$ which leaves $q$ intact but splits
$r$ into two distinct states $r_v$ and $r_{v'}$.
But then security will fail to hold for $A'$, since
$A'(t'\cdot v') = r_v\neq r_{v'} = A'(t\cdot v)$.
So, since we require security to hold for all $A$, we may strengthen
the rhs to $A(t'\cdot v') = A(t\cdot v) \wedge v = v'$.
Then, given $A(t') = A(t)$, since $A$ is a deterministic automaton,
it follows that $v = v' \Rightarrow A(t'\cdot v') = A(t\cdot v)$, hence the rhs simplifies
to just $v = v'$ and the unwinding reduces to:
\[
	\rho \equiv \sigma
	\wedge
	<c,\rho> \overarrowA{t'\cdot v'}{}
	\wedge
	A(t') = A(t)
	\Rightarrow
    v' = v
\]
Finally, since $A$ now only occurs on the lhs,
we see that there is a distinguished counting attacker for which the unwinding
is harder to satisfy than all others, namely the attacker $A_\#$,
for which $A_\#(t') = A_\#(t)$ iff $\length{t'} = \length{t}$.
Thus the property will hold for all $A$ iff it holds for $A_\#$ and so
we arrive at our two-run property:
\begin{definition}[Two-Run PI Security]
\label{def:pisecprop:channel}
Command $c$ is two-run PI secure for policy $D$ on channel $a$ for initial store $\sigma$ if
  whenever
  $<c,\sigma> \overarrowA{t}{n} <c_n,\sigma_n> \overarrowA{v}{1}$ and
  $\rho \equiv \sigma$ and
  $<c,\rho> \overarrowA{t'\cdot v'}{}$ and
  $\length{t'} = \length{t}$,
  then
  $v' = v$,
  where ${\equiv} = D_a(c,\sigma,n)$.
\end{definition}

\begin{theorem}
\label{thm:two-run-completeness}
Let $c$ be a command.
For any given policy, channel and initial store,
$c$ is PI secure against all attackers
iff
$c$ is two-run PI secure.
\end{theorem}
\begin{proof}
This follows from the unwinding of the PI security definition, as shown above.
\end{proof}

\section{A Dependency Type System}
\label{sec:typesystem}

Within the literature on enforcement of
information flow policies,
some work is distinguished by the appearance of explicit dependency analyses.
In the current paper we take as our starting point the flow-sensitive type systems
of \cite{Hunt:Sands:ESOP11,Hunt:Sands:POPL06}, due to the relative simplicity of presentation.
Other papers proposing similar analyses include
\cite{Clark+:JCL},
\cite{Amtoft:Banerjee:SAS04},
\cite{Andrews:Reitman:Axiomatic}
and
\cite{Banatre:Bryce:LeMetayer:ESORICS94}.
Some of the similarities and differences between these approaches
are discussed in \cite{Hunt:Sands:POPL06}.

The original work of \cite{Hunt:Sands:POPL06} defines a family
of type systems,
parameterised by choice of a multi-level security lattice,
and establishes the existence of principal typings within this family.
The later work of \cite{Hunt:Sands:ESOP11} defines a single system which
produces \emph{only} principal types.
In what follows we refer to the particular 
flow-sensitive type system defined in \cite{Hunt:Sands:ESOP11} system as FST.

The typings derived by FST take the form of an
environment $\Gamma$ mapping each program variable ${\tt x}$ to a set $\Gamma({\tt x})$
which has a direct reading as
(a conservative approximation to) the set of dependencies for ${\tt x}$.
All other types derivable using the flow-sensitive type systems of \cite{Hunt:Sands:POPL06} can be
recovered from the principal type derived by FST.
Because principal types are simply dependency sets, they are not specific to any particular
security hierarchy or policy. This is the basis of the clean separation we are able to
achieve between analysis and policy verification in what follows.

\begin{wrapfigure}{r}{0.22\textwidth}
  \vspace{-1em}
  \begin{lstlisting}
  x := z + 1;
  z := x;
  if (z > 0)
    y := 1;
  x := 0;
  \end{lstlisting}
  \caption{}
  \vspace{-1em}
  \label{fig:depex2}
\end{wrapfigure}

Consider the simple program shown in Figure~\ref{fig:depex2}.
The type inferred for this program is $\Gamma$, where
$\Gamma({\tt x}) = \{ \}$,
$\Gamma({\tt y}) = \{ {\tt y}, {\tt z} \}$,
$\Gamma({\tt z}) = \{ {\tt z} \}$.
From this typing we can verify, for example, any static policy
using a security lattice in which
${\it level}({\tt z}) \sqsubseteq {\it level}({\tt y})$.

FST is defined only for a simple language which does not include output statements.
This makes it unsuitable for direct application to verification of dynamic policies,
so in the current paper we describe a straightforward extenion of FST to a language with output statements.
We then show how the inferred types can be used to enforce policies such as
those in \cite{askarov2012} and \cite{Paragon},
which appear very different from the simple static, multi-level policies
originally targeted.

\subsection{Language}
We instantiate the abstract computation model of
Section~\ref{sec:computation} with 
a simple while-language with output channels, shown in
Figure~\ref{fig:semantics}.
We let ${\tt x} \in \PVar$ range over program variables,
$a \in \Chan$ range over channels (as before) and
$p \in \PPoint$ range over program points. Here non-silent output labels have the
form $(a,v,p)$, $\chan(a,v,p) = a$, and $\val(a,v,p) = v$.

 The language is similar to the one
considered in~\cite{askarov2012}, except for the absence of input channels.

Outputs have to be annotated with a program point $p$ to bridge between the
dependency analysis and the policy analysis, described in Section~\ref{sec:semsoundness}.

\begin{figure}[t]
\begin{tabular}{lll}
\text{Values}      & $v ::=$ & $n$ \qquad\qquad
\text{Expressions}  $e ::=$  $v \mid {\tt x}$ \\
\text{Commands}    & $c ::=$ & ${\tt skip} \mid c_1;c_2 \mid {\tt x} := e \mid {\tt if}\ e\ c_1\ c_2$ 
                               $\mid {\tt while}\ e\ c \mid {\tt out}\ e\ {\tt on}\ a\ {\tt @}\ p$ 
\end{tabular}
\[ \inferrule{}
             {<{\tt skip};c,\sigma> \overarrow{\epsilon}{} <c,\sigma>}
   \qquad
   \inferrule{<c_1,\sigma> \overarrow{\alpha}{} <c'_1,\sigma'>}
             {<c_1;c_2,\sigma> \overarrow{\alpha}{} <c_1';c_2,\sigma'>}
   \qquad
   \inferrule{\sigma(e) = v}
             {<{\tt x} := e,\sigma> \overarrow{\epsilon}{} <{\tt skip},\sigma'>}
\]
\[
   \inferrule{\sigma(e) = v}
             {<{\tt out}\ e\ {\tt on}\ a\ {\tt @}\ p,\sigma> \overarrow{(a,v,p)}{} <{\tt skip},\sigma'>}
   \qquad
   \inferrule{}
             {<{\tt while}\ e\ c,\sigma> \overarrow{\epsilon}{}
              <{\tt if}\ e\ (c; {\tt while}\ e\ c)\ {\tt skip}, \sigma>}
\]
\[ \inferrule{\sigma(e) \not= 0}
             {<{\tt if}\ e\ c_1\ c_2,\sigma> \overarrow{\epsilon}{} <c_1,\sigma>}
   \qquad
   \inferrule{\sigma(e) = 0}
             {<{\tt if}\ e\ c_1\ c_2,\sigma> \overarrow{\epsilon}{} <c_2,\sigma>}
\]
\caption{Language and semantics.}
\label{fig:semantics}
\end{figure}

\subsection{Generic typing}

Traditional type systems for information flow
assume that all sensitive inputs to the system (here: program variables)
are associated with a security level.
Expressions in the command to be typed might combine information with different security
levels.
To ensure that all expression can be typed, the security levels are therefore
required to form at least a join-semilattice, or in some cases a full lattice.
The type system then ensures no information of a (combined) level $l_1$ can be
written to a program variable
with level $l_2$ unless $l_1 \sqsubseteq l_2$.

The system FST from Hunt and Sands~\cite{Hunt:Sands:ESOP11} differs from these type systems in two ways.
Firstly, it does not require intermediate assignments
to respect the security lattice ordering. As an observer is assumed to only see the
final state of the program, only the final value of a variable must
not depend on any information which is forbidden by the lattice ordering. For example,
suppose
${\it level}({\tt y}) \sqsubseteq {\it level}({\tt z}) \sqsubseteq {\it level}({\tt x})$
but
${\it level}({\tt x}) \not\sqsubseteq {\it level}({\tt z})$
and consider the first two assignments in the example from Fig.~\ref{fig:depex2}.
\[ \verb!x = z + 1; z = x;! \]
A traditional type system would label this command as insecure because of
the assignment {\tt z = x} and the fact that ${\it level}({\tt x}) \not\sqsubseteq {\it level}({\tt z})$,
even though the value of {\tt z} after this assignment does not depend on the initial value of
{\tt x} at all.
FST however is \emph{flow-sensitive}
and allows the security label on {\tt x} to vary through the code.

Secondly, and more significantly, by using the powerset of program variables as security lattice,
FST  provides a \emph{principal typing} from which
all other possible typings can be inferred.

Thus the typing by FST is generic: a command needs to be typed
only once and can then be verified against any static information-flow policy.
Since the ordering among labels is not relevant while deriving the typing,
FST is also able to verify policies which are not presented in
the shape of a security lattice, but any relational {\it `may-flow'} predicate between
security labels can be verified.

\subsection{Generic typing for dynamic policies}

We now present an extended version of FST which includes an
additional typing rule for outputs.
All the original typing rules of FST remain unchanged.

Intuitively, an output on a channel is like the final assignment to a variable
in the original FST, that is, its value can be observed.
Since types are sets of dependencies, 
we could simply type an output channel as the union of all dependencies resulting
from all output statements for that channel.
This would be sound but unduly imprecise:
the only flows permitted would be those permitted by the policy \emph{at all times},
in effect requiring us to conservatively approximate each dynamic policy by a static one.
But we can do better than this.

The flow-sensitivity of FST means that a type derivation infers types
at intermediate program points which will, in general, be different from
the top-level type inferred for the program.
These intermediate types are not relevant for variables, since their intermediate values are
not observable.
But the outputs on channels at intermediate points
\emph{are} observable, and so intermediate channel types \emph{are} relevant.
Therefore, for each channel we record in $\Gamma$ distinct dependency sets for each program point
at which an output statement on that channel occurs.
Of course, this is still a static approximation of runtime behaviour.
While our simple examples of dynamic policies explicitly associate
policy changes to program points, for real-world use more expressive dynamic policy languages may be needed.
In Section~\ref{sec:dynpolicies} we formally define the semantics of a dynamic policy as an arbitrary function of
a program's execution history, which provides a high degree of generality.
However, in order to apply a typing to the verification of such a policy, it is first necessary
to conservatively approximate the flows permitted by the policy at each program point of interest
(Definition~\ref{def:policyapproximation}).

Let $X$ be the dependency set for the channel-$a$ output statement at program point $p$.
The meaning%
\footnote{This is progress-insensitive dependency (see Section~\ref{sec:knowledgebased}).
A progress-sensitive version can be defined in a similar way.}
of $X$ is as follows:
\begin{quote}
  Let $\sigma$ be a store such that execution starting in $\sigma$ arrives at $p$, producing the
  $i$'th output on $a$.
  Let $\rho$ be any store which agrees with $\sigma$ on all variables in $X$ and also eventually
  produces an $i$'th output on $a$ (not necessarily at the same program point).
  Then these two outputs will be equal.
\end{quote}
Two key aspects of our use of program points should be highlighted:
\begin{enumerate}

   \item While the intended semantics of $X$ as outlined above does not require corresponding outputs
on different runs to be produced at the same program point, the $X$ that is inferred by the
type system \emph{does} guarantee this stronger property. Essentially this is because (in common with all similar
analyses) the type system uses
control-flow dependency as a conservative proxy for the semantic dependency property of interest.

   \item Our choice of program point to distinguish between different ouputs on the same channel is not arbitrary;
it is essentially forced by the structure of the original type system.
As noted, program point annotations simply allow us to record in the final typing
exactly those intermediate dependency sets which are already
inferred by the underlying flow-sensitive system.
While it would be possible in principle to make even finer distinctions
(for example, aiming for path-sensitivity rather than just flow-sensitivity)
this would require fundamental changes to the type system.

\end{enumerate}

\begin{figure}
\begin{align*}
\ruleName{\TSSkip} &\quad
  \inferrule {}
             {\vdash \{ {\tt skip} \}\ \Gammaid}
\\
\ruleName{\TSAssign} &\quad
  \inferrule {}
             { \vdash \{{\tt x :=}\ e\}\ \Gammaid\ [ {\tt x} \mapsto {\it fv}(e) \cup \{ {\tt pc} \}  ]}
\\
\ruleName{\TSSequence} &\quad\hspace{-0.3ex}
  \inferrule {\vdash \{ c_1 \} \Gamma_1  \\ \vdash \{ c_2 \} \Gamma_2}
             {\vdash \{c_1\ {\tt ;}\ c_2\}\ \Gamma_2 ; \Gamma_1}
\end{align*} \vspace{-3ex}
\begin{align*}
\ruleName{\TSIfElse} &
\\ 
\span\inferrule{\vdash \{ c_i \} \Gamma_i \\
              \vdash \Gamma'_i = \Gamma_i ; \Gammaid[{\tt pc} \mapsto \{{\tt pc}\} \cup {\it fv}(e)] \\
              i = 1, 2  }
             {\vdash \{{\tt if}\ e\ c_1\ c_2 \}\ (\Gamma'_1 \cup \Gamma'_2)
   [ {\tt pc} \mapsto \{ {\tt pc} \}
   ]}
\\
\ruleName{\TSWhile} &
\\ 
\span\inferrule{ \vdash \{ c \} \Gamma_c
    \hspace{1cm}
      \Gamma_f = (\Gamma_c ; \Gammaid[\pc \mapsto \{\pc\} \cup \mathit{fv}(e)])^{*}
    }
    {\vdash \{{\tt while}\ e\ c\}\ \Gamma_f\ 
   [ {\tt pc} \mapsto \{ {\tt pc} \}
   ]}
\\
\ruleName{\TSOutput} &
\\
\span\inferrule {}
             {  \vdash \{{\tt out}\ e\ \mathtt{on}\ a\ \mathtt{@}\ p\}  \Gammaid  [ a_p \mapsto {\it fv}(e) \cup \{{\tt pc}, a, a_p \}
   ; a \mapsto \{{\tt pc}, a\}
   ]}
\end{align*}
\caption{Type System.}
\label{fig:typesystem}
\end{figure}

The resulting type system is shown in Figure~\ref{fig:typesystem}.
We now proceed informally to motivate its rules.
Definitions and proofs of formal soundness are presented in Section~\ref{sec:semsoundness}.

The type system derives judgements of the form $|- \{c\} \Gamma$, where
$\Gamma : \Var \rightarrow 2^{\Var}$ is an environment mapping variables
to a set of dependencies. The variables we consider are 
$\Var = \PVar \cup \CPoint \cup \{\pc\} \cup \Chan$
with $\CPoint = \Chan \times \PPoint$.
We consider the relevance of each kind of variable in turn.

\begin{itemize}
  \item As program variables $\PVar$ form the inputs to the command, these are
        the dependencies of interest in the typing of a command. For
        program variables themselves, $\Gamma({\tt x})$ are the dependencies for
        which a different intial value might result in a different final value of {\tt x}.
  \item Pairs of channels and program points $(a,p) \in \CPoint$ are denoted as
        $a_p$. The dependencies $\Gamma(a_p)$ are those program variables
        for which a difference in initial value might cause a difference
        in the value of any observation that can result from an output statement for channel $a$ with annotation $p$.
  \item Whenever the program counter $\pc \mathop{\in} \Gamma({\tt x})$ this indicates that this command
        potentially changes the value of program variable {\tt x}. Similar, if $\pc \mathop{\in} \Gamma(a)$
        then $c$ might produce an output on channel $a$ and if $\pc \mathop{\in} \Gamma(a_p)$
        then $c$ might produce an output on $a$ caused by a statement annotated with $p$.
        We use the program counter to catch implicit flows that may manifest
        in these ways.
  \item We use $\Chan$ to capture the latent flows described in example program B
        in the introduction.
        The dependencies $\Gamma(a)$ are those program variables
        for which a difference in initial value might result in a different
        number of outputs produced on channel $a$ by this command.
        This approach to address latent flows was first introduced in~\cite{askarov2012} as
        \emph{channel countext bounds}.
\end{itemize}
We first explain the notation used in the unchanged rules from FST
before turning our attention to the new \ruleName{\TSOutput} rule. All concepts
have been previously introduced in~\cite{Hunt:Sands:ESOP11}.

The function ${\it fv}(e)$ returns the free variables in expression $e$.
The identity environment
$\Gammaid$ maps each variable to the singleton set of itself, that is
$\Gammaid(x) \mathop{=} \{x\}$ for all $x \mathop{\in} \Var$.
Sequential composition of environments is defined as:
\[\Gamma_2;\Gamma_1(x) = \bigcup_{y \in \Gamma_2(x)}\Gamma_1(y)\]
Intuitively, $\Gamma_2;\Gamma_1$ is as $\Gamma_2$ but substituting the dependency relations
already established in $\Gamma_1$.
We overload the union operator for environments:  $(\Gamma_1 \cup \Gamma_2)(x) = \Gamma_1(x) \cup \Gamma_2(x)$.
We write $\Gamma^{*}$ for the
fixed-point of $\Gamma$, used in \ruleName{\TSWhile}:
\[
\Gamma^{*} = \bigcup_{n \geq 0} \Gamma^n \qquad
\textnormal{ where } \Gamma^0 = \Gammaid
\textnormal{ and } \Gamma^{n+1} = \Gamma^n; \Gamma
\]

It is only in the typing \ruleName{\TSOutput} of the output command that the additional channel and
program point dependencies are mentioned; this underlines our statement that
extending FST to target dynamic policies is straightforward.

We explain the changes to $\Gammaid$ in \ruleName{\TSOutput} in turn.
For $a_p$, clearly the value of the output and thus the observation is affected
by the program variables occuring in the expression $e$.
We also include the program counter $\pc$ to catch implicit flows;
if we have a command of the form 
 ${\tt if}\ e\ ({\tt out}\ 1\ {\tt on}\ a\ @\ p)\ ({\tt out}\ 2\ {\tt on}\ a\ @\ q)$
the value of the observation caused by output $a_p$ is affected by the
branching decision, which is caught in \ruleName{\TSIfElse}.

We include the channel context bounds $a$ for the channel
on which this output occurs to capture the latent flows of earlier conditional
outputs, as demonstrated in the introduction.
Observe that by the definition of sequential composition of environments,
we only add those dependencies for conditional outputs that happened
\emph{before} this output. That is, not the ones that follow this output, since
it cannot leak information about the absence of future observations.

Finally, we include the dependencies of output point $a_p$ itself. 
By doing so the dependency set of $a_p$ becomes \emph{cumulative}:
with every sequential composition (including those used in $\Gamma^{*}$)
the dependency set of $a_p$ only grows, as opposed to 
the dependencies of program variables.
This makes us sum the dependencies of all outputs on channel $a$ annotated with
the same program point, as we argued earlier.

The mapping for channel context bounds $a$ is motivated in a similar manner.
The $\pc$ is included since the variables affecting whether this output occurs
on channel $a$ are the same as those that affect whether this statement is reached.
Note that we are over-approximating here, as the type system adds the dependencies of
$e$ in \[
{\tt if}\ e\ ({\tt out}\ 1\ {\tt on}\ a\ @\ p_1)\ ({\tt out}\ 2\ {\tt
  on}\ a\ @\ p_2)
\]
to context bounds $a$, even though the number of outputs is always one.

Like for $a_p$, we make $a$ depend on itself, thus accumulating all the
dependencies that affect the number of outputs on channel $a$.

As the \ruleName{\TSOutput} rule does not introduce more complex operations than
already present, the type system has the same complexity as FST. That is,
the type system can be used to construct a generic type in $O(nv^3)$
where $n$ is the size of the program and $v$ the number of variables in \Var.

\section{Semantic Soundness and Policy Compliance}
\label{sec:semsoundness}
We present a soundness condition for the
type system, and show that the type system is sound.
We then describe how the generic typings that the type system derives can be used to check
compliance with a dynamic policy that is approximated per program point.
We begin by showing how an equivalence relation on stores can be created from
a typing:
\begin{definition}
  We write $\Gequiv{x}$ for the equivalence relation corresponding to the typing $\Gamma$ of
  variable $x \in \Var$, defined as $\sigma \Gequiv{x} \rho$ iff
  $\sigma({\tt y}) \mathop{=} \rho({\tt y})$ for all ${\tt y} \in \Gamma(x)$.
\end{definition}

As we are using $\Gamma(a_p)$ as the approximation of dependencies for an
observation, the soundness of the PI type system is similar to the PI
security for dynamic policies, except that we take the equivalence relation
as defined by $\Gamma(a_p)$ rather than the policy~$D$.

\begin{definition}[PI Type System Soundness]
\label{def:pitssound}
We say that the typing $|- \{c\} \Gamma$ is sound iff
for all $\sigma, \rho$, if
  $<c,\sigma> \overarrowA{t}{*} <c_{\sigma},\sigma'> \overarrow{(a,v,p)}{}$ and
  $<c,\rho> \overarrowA{t'}{*} <c_{\rho},\rho'> \overarrowA{v'}{}$ and
  $\length{t} = \length{t'}$ then
  $\sigma \Gequiv{a_p} \rho  \Rightarrow v \mathop{=} v'$.
\end{definition}

\begin{theorem}
\label{the:pisoundness}
All typings derived by the type system are sound.
\end{theorem}

The proof for Theorem~\ref{the:pisoundness} can be found in \paperTR{Appendix A of~\cite{delft2015TR}}{Appendix~\ref{app:pits}}.

To link the typing and the actual dynamic policy, we rely
on an analysis that is able to approximate the dynamic policy per program point.
A sound approximation should return a policy that is at least as restrictive
as the actual policy for any observation on that program point.

\begin{definition}[Dynamic Policy Approximation]
\label{def:policyapproximation}
A dynamic policy approximation
$A : \CPoint \rightarrow 2^{\Sigma \times \Sigma}$ is a mapping from channel and
program point pairs to an equivalence relation on stores. 
The approximation $A$ on command $c$, written $c:A$,
is sound for dynamic policy $D$ iff, for all $\sigma$ if 
$<c,\sigma>\overarrow{}{n}<c',\sigma'>\overarrow{(a,v,p)}{}$ then
$A(a_p)$ is coarser than $D(c,\sigma,n)$.
\end{definition}

We now arrive at the main theorem in this section.
Given a typing $|- \{c\} \Gamma$, we can now easily verify for command $c$ its
compliance with \emph{any} soundly approximated dynamic policy, by
simply checking that the typing's policy is at least as restrictive as the
approximated dynamic policy for every program point.

\begin{theorem}[PI Dynamic Policy Compliance]
\label{the:pidyncompl}
Let $c\mathop{:}A$ be a sound approximation of dynamic policy $D$. If
$|-\{c\} \Gamma$ and $\Gequiv{a_p}$ is coarser than $A(a_p)$ for all
program points $a_p$, then $c$
is two-run PI secure for $D$ on all channels and for all initial stores.
\end{theorem}
\begin{proof}
  Given a store $\sigma$ such that $<c,\sigma> \overarrowA{t}{n} <c_{\sigma},\sigma'> \overarrow{(a,v,p)}{}$ and
  a store $\rho$ such that $<c,\rho> \overarrowA{t'}{*} <c_{\rho},\rho'> \overarrowA{v'}{}$  and
  $\length{t} = \length{t'}$ and
  $\sigma D_a(c,\sigma,n) \rho$, we need to show that $v \mathop{=} v'$.
  Since $c\mathop{:}A$ is a sound approximation of $D$, we have that $\sigma A(a_p) \rho$
  and as $\Gequiv{a_p}$ is coarser than $A(a_p)$ we also have $\sigma \Gequiv{a_p} \rho$.
  Which by Theorem~\ref{the:pisoundness} gives us that $v \mathop{=} v'$.
\end{proof}
\begin{corollary}
If the conditions of Theorem~\ref{the:pidyncompl} are met, then
$c$ is PI secure for D for all attackers. This is immediate by
Theorem~\ref{thm:two-run-completeness}.
\end{corollary}

\section{Related Work}
\label{sec:related}
In this section we consider the related work
on security for dynamic policies and
on generic enforcement mechanisms for information-flow control.
We already discuss the knowledge-based definitions by Askarov and Chong~\cite{askarov2012}
in detail in Section~\ref{sec:knowledgebased}. 

The generality of expressing dynamic policies per execution point can be
identified already in the early work by Goguen and Meseguer~\cite{Goguen:Meseguer:Noninterference}.
They introduce the notion of conditional noninterference as a noninterference
relation that should hold per step in the system, provided that some
condition on the execution history holds. Conditional noninterference
has been recently revisited by Zhang~\cite{Zhang2012} who uses unwinding relations to
present a collection of properties that can be verified by existing proof assistants.

Broberg and Sands \cite{Broberg:Sands:PLAS09} developed another
knowledge-based definition of security for dynamic
policies which only dealt with the attacker with
perfect recall. The approach was specialised to the specific dynamic
policy mechanism Paralocks~\cite{Broberg:Sands:Paralocks} which uses part of the program state to vary the ordering between security labels.  

Balliu et al.~\cite{Balliu:2011:ETL} introduce a temporal epistemic logic to
express information flow policies. Like our dynamic policies, the epistemic formulas
are to be satisfied per execution point, suggesting that we are able express a
similar set of dynamic policies. Dynamic policies can be individually checked
by the \textsc{ENCoVer} tool~\cite{balliu2012encover}.

The way in which we define dynamic policies matches exactly the set of
synchronous dyanmic policies: those policies that deterministically determine
the active policy based on an execution point.
Conversely, an asynchronously changing policy cannot be deterministically determined
from an execution point, but is influenced by an environment external to the
running program.

There is relatively little work on the enforcement of asynchronous dynamic policies.
Swamy et al.~\cite{Swamy+:Managing} present the language {\sc Rx} where policies are
define in a role-based fashion, where membership and delegation of roles can 
change dynamically. Hicks et al.~\cite{Hicks+:Dynamic} present an extension to
the DLM model, allowing the acts-for hierarchy among principals to change while
the prorgam is running.

Both approaches however need a mechanism to synchronise the policy changes with
the program in order to enforce information-flow properties. {\sc Rx} uses transactions
which can rollback when a change in policy violates some of the flows in it,
whereas the work by Hick et al. inserts automatically derived coercions that
force run-time checks whenever the policy changes.

One of the beneficial characteristics of our enforcement approach is that 
commands need to be analysed only once to be verified against multiple
information-flow policies. This generality can also be found in the work
by Stefan et al.~\cite{Stefan+:Flexible} presenting LIO, a Haskell library for
inforamtion-flow enforcement which is also parametric in the security policy.
The main differences between our approach and theirs is that 
LIO's enforcement is dynamic rather than static,
while the enforced policies are static rather than dynamic.

\section{Conclusions}
\label{sec:conclusions}

We extended the flow-sensitive type system from~\cite{Hunt:Sands:POPL06}
to provide for each output channel individual dependency sets per point in
the program and demonstrated that this is sufficient to support dynamic
information flow policies. We proved the type system sound with respect to
a straightforward two-run property which we showed sufficient to imply
knowledge-based security conditions.

As our approach allows for the core of the analysis to be performed
even before the policy is known, this enables us to reuse the results of the
dependency analysis across the verification of multiple types of policies.
An interesting direction for future research could be on the
possibility to use the dependency analyses performed by advanced slicing
tools such as JOANA~\cite{JOANA} and Indus~\cite{Indus}.


\bibliographystyle{alpha}
\bibliography{literature}{}

\paperTR{}{
  \appendix

  \section{Type System Soundness}
\label{app:pits}

We need a collection of lemmas on the type system, which we prove correct in
\paperTR{the extended version of our paper.}
        {appendix~\ref{app:auxlemmas}.}
        
\begin{lemma}[Preserving Dependencies]
\label{app:lemma:paper:presdep}
The dependencies of program points only
increase with sequential composition.
That is, if $|- \{ c_1 \} \Gamma_1$ and $|- \{c_1;c_2\} \Gamma$ then $\Gamma_1(a_p) \subseteq \Gamma(a_p)$
for all $a_p$.
\end{lemma}

\begin{lemma}[Non-Branching Reduction]
\label{app:lemma:paper:nbr}
For all configurations $c$ and stores $\sigma, \rho$ such that
  $|- \{c\} \Gamma$ with
  $<c,\sigma> \overarrow{\alpha}{} <c',\sigma'>$ and
  $<c,\rho> \overarrow{\beta}{} <c',\rho'>$,
  i.e.\ to the same command $c'$,
  where $\alpha, \beta$ is either an output or $\epsilon$. Then
  $\sigma \Gequiv{a_p} \rho$ implies
  that $|- \{c'\} \Gamma'$ and
  $\sigma' \NGequiv{\Gamma'}{a_p} \rho'$.
\end{lemma}

\begin{lemma}[Branching Reduction]
\label{app:lemma:paper:br}
For all configurations $c$ and stores $\sigma, \rho$such that
  $|- \{c\} \Gamma$ with
  $<c,\sigma> \overarrow{\epsilon}{} <c'_\sigma,\sigma'>$ and
  $<c,\rho> \overarrow{\epsilon}{} <c'_\rho,\rho'>$,
  where $c'_\sigma \mathop{\not=} c'_\rho$. 
  For all $a_p$ such that
  $<c'_\sigma,\sigma'> \overarrowA{t(a,v,p)}{*}$, it holds that 
  if $\sigma \Gequiv{a_p} \rho$ then either:
\begin{itemize}
  \item There exists a joining command $c_j$ with $t_1 \mathop{\preceq} t$,
        such that
        $<c'_\sigma,\sigma'> \overarrowA{t_1}{*} <c_j,\sigma_j>$
        and
        $<c'_\rho,\rho'> \overarrowA{t'_1}{*} <c_j,\rho_j>$
        with $\length{t_1} \mathop{=} \length{t'_1}$ and
        $|- \{c_j\} \Gamma_j$ with $\sigma_j \NGequiv{\Gamma_j}{a_p} \rho_j$.
        I.e. both executions join again at equal command $c_j$ after an equal number
        of outputs with equivalent dependencies for $a_p$, or
  \item For all $t'$ such that 
        $<c,\rho> \overarrowA{t'}{*}$, $\length{t} \geq \length{t'}$.
\end{itemize}
\end{lemma}

\noindent
We restate the semantic soundness property for the progress-insensitive type
sytem:
\begin{align}
\label{app:eq:pits:ih}
\begin{minipage}{24em}
 For all $n \geq 0$, for all commands $c$ with typing $|- \{c\} \Gamma$, we have
for all stores $\sigma, \rho$, if
  $<c,\sigma> \overarrowA{t}{n} <c_{\sigma},\sigma'> \overarrow{(a,v,p)}{}$ and
  $<c,\rho> \overarrowA{t'}{*} <c_{\rho},\rho'> \overarrowA{v'}{}$ and
  $\length{t} = \length{t'}$ then
  $\sigma \Gequiv{a_p} \rho  \Rightarrow v \mathop{=} v'$.
\end{minipage} 
\end{align}

\noindent
We show this by complete (strong) induction on $n$.

\paragraph{$(n \mathop{=} 0)$}
  When $c$ produces an output in a single step for some store, it does so for
  any store and at the same program point (Fig.~\ref{fig:semantics}). We need to show
  $$
    <c,\sigma> \overarrow{(a,v,p)}{} \and \sigma \Gequiv{a_p} \rho \and
    <c,\rho> \overarrow{(a,v',p)}{} 
    ==>
      v = v'
  $$
  By induction on the step taken from $<c,\sigma>$. Due to the output, there are
  only two cases that do not lead to immediate contradiction:
  \begin{itemize}
    \item ${\tt out}\ e\ {\tt on}\ a\ {\tt @}\ p$ -- 
      By \ruleName{\TSOutput},
      ${\it fv}(e) \subseteq \Gamma(l,p)$. Therefore,
      $\sigma(e) = \rho(e)$ and thus $v = v'$.
    \item $c_1; c_2$ -- By Lemma~\ref{app:lemma:paper:presdep},
      $\Gamma_1(l,p) \subseteq \Gamma(l,p)$ for $|- \{c_1\}\Gamma_1$.
      Therefore $\sigma \NGequiv{\Gamma_1}{a_p} \rho$ and the
      property follows by induction on $<c_1,\sigma>$.
  \end{itemize}
  
\paragraph{\mathligsoff$(n + 1 > 0)$\mathligson}
We have $<c,\sigma> \overarrowA{\alpha}{} <c'_\sigma,\sigma'> \overarrowA{t}{n} <c''_{\sigma},\sigma''> \overarrow{(a,v,p)}{}$
and $<c,\rho> \overarrowA{\beta}{} <c'_\rho,\rho'>$ $\overarrowA{t'}{*} < c''_{\rho},\rho''> \overarrowA{v'}{}$
with the induction hypothesis \eqref{app:eq:pits:ih} for evaluations of length $\leq n$.
Here $\alpha$ could be either silent or an output. We only consider the
case where $\alpha$ is silent; for the case where $\alpha$ is an output
the proof is similar except in the induction step we have that traces 
produced by $\sigma'$ and $\rho'$ are both 1 shorter in length.

When $c$ produces no output in a single step for some store, it does so for
  any store (Fig.~\ref{fig:semantics}).
We case split on $c'_\sigma = c'_\rho$:

\hfill\begin{minipage}{\dimexpr\columnwidth-3em}

\begin{itemize}

  \item[$c'_\sigma = c'_\rho$]
    By Lemma~$\ref{app:lemma:paper:nbr}$, $|- \{c'_\sigma\} \Gamma'$ and
    $\sigma' \NGequiv{\Gamma'}{a_p} \rho'$. The case follows by induction on $n$.

  \item[$c'_\sigma \not= c'_\rho$]
    Since $<c,\rho>$ produces a trace longer than $\length{t}$,
    by Lemma~$\ref{app:lemma:paper:br}$
    there exists a $t_1 \mathop{\preceq} t$ such that
    $<c'_\sigma,\sigma'> \overarrowA{t_1}{n_1} <c_j,\sigma_j> \overarrowA{t_2}{n_2} <c''_{\sigma},\sigma''> \overarrow{(a,v,p)}{}$
    and there is a $t'_1 \mathop{\preceq} t'$ such that
    $<c'_\rho,\rho'> \overarrowA{t'_1}{*} <c_j,\rho_j> \overarrowA{t'_2v'}{*}$
    with $\length{t'_1} = \length{t_1}$ and
    $|- \{c_j\} \Gamma_j$ with $\sigma_j \NGequiv{\Gamma_j}{a_p} \rho_j$.

    Since $\length{t} = \length{t'}$ and $\length{t_1} = \length{t'_1}$ we 
    have $\length{t_2} = \length{t'_2}$. Then the case follows by induction
    on some $n_2 \leq n$.
\end{itemize}
\xdef\tpd{\the\prevdepth}
\end{minipage}

  \section{Proofs of auxiliary lemmas}
  \label{app:auxlemmas}
  We proof our auxiliary lemmas as part of a series of necessary properties of the
type sytem. In this section, the Lemmas
\ref{app:lemma:paper:presdep},
\ref{app:lemma:paper:nbr} and
\ref{app:lemma:paper:br} can be found with proof as Lemmas
\ref{app:lemma:presdep},
\ref{app:lemma:nbr} and
\ref{app:lemma:br} respectively.

\begin{lemma}
  \label{app:lemma:pcunchanged}
  For all $|- \{c\} \Gamma$, we have $\Gamma(\pc) \mathop{=} \{\pc\}$.
\end{lemma}

\begin{proof}
  By induction on the typing derivation.
  This property holds clearly for \ruleName{\TSSkip},
  \ruleName{\TSAssign} and \ruleName{\TSOutput} as $\Gamma(\pc) =$ 
  $\Gammaid(\pc) \mathop{=} \{\pc\}$.
  It is also clear for \ruleName{\TSIfElse} and \ruleName{\TSWhile} as they
  explicitly reset the dependencies of $\pc$ to $\{\pc\}$.
  For \ruleName{\TSSequence}, by induction we have $\Gamma_2(\pc) \mathop{=} \{\pc\}$
  thus $\Gamma_2;\Gamma_1(\pc) = \Gamma_1(\pc)$, which by induction is $\{\pc\}$.
\end{proof}

\begin{lemma}
  \label{app:lemma:nopcid}
  For all $|- \{c\} \Gamma$, if $\pc \not\in \Gamma(x)$ then $\Gamma(x) = \Gammaid(x) = \{x\}$, for
  all variables, channels and program points $x$.
\end{lemma}

\begin{proof}
  By induction on the derivation of $|- \{c\} \Gamma$:

  \begin{itemize}

  \item Case: \ruleName{\TSAssign}

  $\Gamma$ is different from $\Gammaid$ only for variable {\tt x}.
  For {\tt x}, $\pc \in \Gamma({\tt x})$ so the property holds trivially.

  \item Case: \ruleName{\TSOutput}

  $\Gamma$ is different from $\Gammaid$ only for channel $a$ and program point $a_p$.
  For them, $\pc \in \Gamma(a)$ and $\pc \in \Gamma(a_p)$ so the property holds trivially.

  \item Case: \ruleName{\TSSequence}

  By Lemma~\ref{app:lemma:pcunchanged},
  $\Gamma_1(\pc) = \Gamma_2(\pc) = \{\pc\}$.
  Thus $\pc \in \Gamma_2(x)$ implies $\pc \in \Gamma_2;\Gamma_1(x)$.
  Reversely, $\pc \not\in \Gamma_2;\Gamma_1(x)$ implies $\pc \not\in \Gamma_2(x)$.
  By induction we then have $\Gamma_2(x) = \{x\}$.
  Thus $\Gamma_2;\Gamma_1(x) = \Gamma_1(x)$ and we have $\pc \not\in \Gamma_1(x)$.
  Therefore by induction $\Gamma_1(x) = \Gammaid(x)$ which
  gives $\Gamma_2;\Gamma_1(x) = \{x\}$.

  \item Case: \ruleName{\TSIfElse}

  The composition with  $\Gammaid[\pc \mathop{\mapsto} \{\pc\} \mathop{\cup} {\it fv}(e)]$
  only adds dependencies and does not remove the dependency on $\pc$.
  So $\pc \in \Gamma_i(x)$ implies $\pc \in \Gamma'_i(x)$.
  $\Gamma$ is the union of both typing, so $\pc \in \Gamma'_i(x)$ 
  implies $\pc \in \Gamma(x)$.
  Reversely, $\pc \not\in \Gamma(x)$ implies $\pc \not\in \Gamma_i(x)$.
  By induction, $\Gamma_i(x) = \{x\}$, thus $\Gamma'_i(x) = \{x\}$,
  thus $\Gamma(x) = \{x\}$.

  \item Case: \ruleName{\TSWhile}

  The composition with  $\Gammaid[\pc \mathop{\mapsto} \{\pc\} \mathop{\cup} {\it fv}(e)]$
  only adds dependencies and does not remove the dependency on $\pc$.
  So $\pc \in \Gamma_c(x)$ implies $\pc \in \Gamma_c;\Gammaid[\pc \mathop{\mapsto} \{\pc\} \mathop{\cup}$ ${\it fv}(e)](x)$.
  The reflexive transitive closure operation takes the union over all $\Gamma^n$,
  so $\pc \in \Gamma_c(x)$ implies $\pc \in \Gamma_f(x)$.
  Reversely, $\pc \not\in \Gamma_f(x)$ implies $\pc \not\in \Gamma_c(x)$.
  By induction, $\Gamma_c(x) = \{x\}$, hence neither the composition
  with $\Gammaid[\pc \mathop{\mapsto} \{\pc\} \mathop{\cup} {\it fv}(e)]$ nor the
  closure operation modifies the dependencies of $x$, and $\Gamma_f(x) \mathop{=} \{x\}$.

  \end{itemize}

\end{proof}

\begin{lemma}
  \label{app:lemma:pcnoadd}
  If $\pc \not\in \Gamma(x)$ with $|- \{ c \} \Gamma$ and
  $<c,\sigma> \overarrow{\alpha}{} <c',\sigma'>$, then
    $\pc \not\in \Gamma'(x)$ with $|- \{ c' \} \Gamma'$, for
    all $x$ in variables, channels and program points.
\end{lemma}

\begin{proof}
  By induction on the structure of $c$:

  \begin{itemize}

  \item Case: ${\tt skip}$ 

  Holds trivially as there is no progress.

  \item Case: $c_1;c_2$

  $\Gamma = \Gamma_2;\Gamma_1$ with $|- \{c_1\} \Gamma_1$ and $|- \{c_2\} \Gamma_2$.
  If $\pc \not\in \Gamma(x)$ then by Lemma~\ref{app:lemma:pcunchanged}
  we have $\pc \not\in \Gamma_2(x)$.
  By Lemma~\ref{app:lemma:nopcid} we have $\Gamma_2;\Gamma_1(x) = \{x\}$;
  $\Gamma_2(x) = \{x\}$ and
  thus $\pc \not\in \Gamma_1(x)$. The case follows by induction on $c_1$.

  \item Case: ${\tt x} := e$

  Becomes {\tt skip} after evaluation. For all
  variables {\tt y} other than {\tt x}, $\pc \not\in \Gamma({\tt y})$ and also
  $\pc \not\in \Gamma'({\tt y})$ as $\Gamma'({\tt y}) = \Gammaid({\tt y})$.
  Similar for program points and channels

  \item Case: ${\tt out}\ e\ {\tt on}\ a\ {\tt @}\ p$ 

  $\Gamma(x) = \Gammaid(x)$
  for all variables, channels and program points other than
  $a$ and $a_p$ so the property holds. For $a$ and $a_p$, we have
  $\pc \in \Gamma(a)$ and $\pc \in \Gamma(a_p)$ so holds trivially.

  \item Case: ${\tt if}\ e\  c_1\  c_2$

  $\pc \not\in \Gamma(x)$ implies $\pc \not\in \Gamma_1(x)$ and
  $\pc \not\in \Gamma_2(x)$; see proof of Lemma~\ref{app:lemma:nopcid}. Thus 
  property holds for  both $\sigma(e) \not= 0$ $(\Gamma' = \Gamma_1)$ and $\sigma(e) = 0$ $(\Gamma' = \Gamma_2)$.

  \item Case: ${\tt while}\ e\ c_1$ 

  If $\pc \not\in \Gamma(x)$
  then $\pc \not\in \Gamma_{1}(x)$ with $|- \{c_1\} \Gamma_{1}$; see proof of Lemma~\ref{app:lemma:nopcid}.
  Thus $\Gamma_{1}(x) = \Gamma(x) = \{x\}$ by Lemma~\ref{app:lemma:nopcid}
  and $\Gamma;\Gamma_{1}(x) = \{x\}$. By inspecting the typing of $c'$
  it follows that
  $\Gamma'(x) = (\Gamma;\Gamma_{1};\Gammaid[\pc \mapsto \{\pc\} \cup {\it fv}(e)] \cup \Gammaid)(x) = \{x\}$,
  thus $\pc \not\in \Gamma'(x)$.

  \end{itemize}

\end{proof}

\begin{lemma}
  \label{app:lemma:whileunfold}
  Let $|- \{{\tt while}\ e\ c\} \Gamma$ and
  $|- \{ {\tt if}\ e\ (c; {\tt while}\ e\ c)$ $({\tt skip})\} \Gamma'$.
  Then $\Gamma'(x) \subseteq \Gamma(x)$ for all $x$ in variables, channels and program points.
\end{lemma}

\begin{proof}

  Here, $\Gamma = \Gamma_f[\pc \mapsto \{\pc\}]$, with 
  $\Gamma_f \mathop{=} (\Gamma_c ; \Gammaid[\pc \mapsto \{\pc\} \cup \mathit{fv}(e)])^{*}$.
  We refer to the individual typings collected by the union in $\Gamma_f$
  with $\Gamma^n$.

  Let $|- \{c\} \Gamma_c$. We need to show that $\Gamma'(x) \subseteq \Gamma(x)$
  which by substituting $\Gamma'$ is:
  \begin{align*}
   \big((\Gamma;\Gamma_c);\Gammaid[\pc \mapsto \{\pc\} \cup {\it fv}(e)]& \\
  \cup \ \Gammaid;\Gammaid[\pc \mapsto \{\pc\} \cup {\it fv}(e)]&\big)(x) \subseteq \Gamma(x)
  \end{align*}
  Observing that $\Gammaid;\Gamma \mathop{=} \Gamma$, $x \mathop{\not=} \pc$ and that
  $\Gammaid(x) \mathop{\subseteq} \Gamma(x)$ since $\Gamma = (\Gamma_c;\Gammaid[\pc \mapsto \{\pc\} \cup {\it fv}(e)])^{*}$, this becomes:
  $$\big((\Gamma;\Gamma_c);\Gammaid[\pc \mapsto \{\pc\} \cup {\it fv}(e)]\big)(x) \subseteq \Gamma(x)
  $$
  If $y \in \Gamma(x)$, this means there exists a $\Gamma^n$ such that $y \in \Gamma^n(x)$.
  Thus, $\big(\Gamma_c;\Gammaid[pc \mapsto \{\pc\} \cup {\it fv}(e)]\big)(y) \subseteq
  \Gamma^{n+1}(x) \subseteq \Gamma(x)$. As $\Gamma_c(y) \subseteq \big(\Gamma_c;\Gammaid[pc \mapsto \{\pc\} \cup {\it fv}(e)]\big)(y)$
  we have that for all $y \in \Gamma(x)$, $\Gamma_c(y) \subseteq \Gamma(x)$.
  Thus $\Gamma;\Gamma_c(x) \subseteq \Gamma(x)$,
  so the relation we need to show is
  implied by
  $$\big(\Gamma;\Gammaid[\pc \mapsto \{\pc\} \cup {\it fv}(e)]\big)(x) \subseteq \Gamma(x)
  $$
  (If we want to be more precise, it holds that
  $(\Gamma;\Gamma_c\cup\Gammaid)(x) = \Gamma(x)$, so this relation is in fact
  equivalent to the one we started out with.)
  If $\pc \in \Gamma(x)$, then there exists a $\Gamma^n$ such that $\pc \in \Gamma^n(x)$.
  As by Lemma~\ref{app:lemma:pcunchanged}, $\Gamma_c(\pc) = \{\pc\}$, it follows that
  $\Gamma_c;\Gammaid[\pc \mapsto \{\pc\} \cup {\it fv}(e)](\pc) = \{\pc\} \cup {\it fv}(e)$.
  Thus since $\Gamma^{n+1} = \Gamma^n;(\Gamma_c;\Gammaid[\pc \mapsto \{\pc\} \cup {\it fv}(e)](\pc))$,
  it follows that $\{\pc\} \cup {\it fv}(e) \subseteq \Gamma^{n+1}(x)$
  and therefore $\{\pc\} \cup {\it fv}(e) \subseteq \Gamma(x)$. So the relation simplifies to
  $\Gamma(x) \subseteq \Gamma(x)$ which holds trivially.

\end{proof}

\begin{lemma}[Program Counter]
  \label{app:lemma:pc}
  The typing $|- \{c\} \Gamma$ guarantees that:
  \begin{enumerate}
    \item \label{app:lemma:pc:vars} If ${\tt pc} \not\in \Gamma({\tt x})$
          and $<c,\sigma> \overarrow{}{*} <{\tt skip},\sigma'>$
          then $\sigma({\tt x}) = \sigma'({\tt x})$.
    \item \label{app:lemma:pc:channels} If ${\tt pc} \not\in \Gamma(a)$
          and $<c,\sigma> \overarrowA{t}{*}$
          then $t = \epsilon$.
    \item \label{app:lemma:pc:points} If ${\tt pc} \not\in \Gamma(a_p)$
          and $<c,\sigma> \overarrowA{t}{*}$
          then $a_p \not\in t$.
  \end{enumerate}
\end{lemma}

\begin{proof}
  Property \ref{app:lemma:pc:vars}
  by induction on $n$ in $<c,\sigma> \overarrow{}{n} <{\tt skip}, \sigma'>$. If $n = 0$
  the property holds trivially as $\sigma = \sigma'$. We thus have:
  $<c,\sigma> \overarrow{\alpha}{} <c',\sigma'> \overarrow{}{n} <{\tt skip}, \sigma''>$
  and by induction if $\pc \not\in \Gamma'({\tt x})$ with $|- \{c'\}\Gamma'$, then
  $\sigma'({\tt x}) = \sigma''({\tt x})$.
  By Lemma~\ref{app:lemma:pcnoadd}, $\pc \not\in \Gamma({\tt x})$ then
  $\pc \not\in \Gamma'({\tt x})$ so we can use the induction hypothesis.

  Similar for properties \ref{app:lemma:pc:channels} and \ref{app:lemma:pc:points}.
\end{proof}

\begin{lemma}[Variables]
  \label{app:lemma:vars}
  For all $c, \sigma, \rho$ such that
    $|- \{c\} \Gamma$ with
    $<c,\sigma> \overarrow{}{*} <{\tt skip},\sigma'>$ and
    $<c,\rho> \overarrow{}{*} <{\tt skip},\rho'>$, then
    $\sigma \Gequiv{\tt x} \rho$ implies
    $\sigma'({\tt x}) = \rho'({\tt x})$.
\end{lemma}

\begin{proof}
  By complete induction on $n$ in $<c,\sigma> \overarrow{}{n} <{\tt skip},\sigma'>$.
  If $n \mathop{=} 0$ the property holds trivially as $\sigma = \sigma'$, $\rho = \rho'$,
  $\Gamma = \Gammaid$ and thus $\rho({\tt x}) = \sigma({\tt x})$.
  We then have $<c,\sigma> \overarrow{\alpha}{}$ $<c',\sigma'> \overarrow{}{n} <{\tt skip},\sigma''>$
  and $<c,\rho> \overarrow{\beta}{} <c'',\rho'> \overarrow{}{*} <{\tt skip},\rho''>$.
  By induction if $\sigma' \NGequiv{\Gamma'}{\tt x} \rho'$ with $|- \{c'\} \Gamma'$ then
  $\sigma''({\tt x}) = \rho''({\tt x})$.

  By cases on the first evaluation step:

  \begin{itemize}

  \item $< {\tt skip}; c, \sigma > \overarrow{\epsilon}{} < c, \sigma >$.

  $\Gamma = \Gamma';\Gammaid$ so $\sigma \Gequiv{\tt x} \rho$
  implies $\sigma \NGequiv{\Gamma'}{\tt x} \rho$ and as 
  $\sigma = \sigma'$ and $\rho = \rho'$ the case follows by induction.

  \item $< c_1; c_2, \sigma > \overarrow{\alpha}{} < c_1'; c_2, \sigma' >$
  with $< c_1, \sigma > \overarrow{\alpha}{} < c_1', \sigma' >$.

  $\Gamma = \Gamma_2;\Gamma_1$. By definition  of $;$ this means that
  $\sigma \Gequiv{\tt x} \rho$ implies that
  $\sigma \NGequiv{\Gamma_1}{\tt y} \rho$ for all ${\tt y} \in \Gamma_2({\tt x})$.
  Thus in the computations
  $<c_1;c_2,\sigma> \overarrow{}{*} <{\tt skip};c_2, \sigma'> \overarrow{}{*} <{\tt skip},\sigma''>$
  and
  $<c_1;c_2,\rho> \overarrow{}{*}$ $<{\tt skip};c_2, \rho'> \overarrow{}{*} <{\tt skip},\rho''>$
  we have by induction on $c_1$ that $\sigma'({\tt y}) = \rho'({\tt y})$. Therefore
  $\sigma' \NGequiv{\Gamma_2}{\tt x} \rho'$  and the case follows on the shorter
  execution trace produced by $c_2$.

  \item $< {\tt x} := e, \sigma > \overarrow{\epsilon}{} < {\tt skip}, \sigma[{\tt x} \mapsto v] >$
  with $\sigma(e) = v$.

  For {\tt x}, we have $\mathit{fv}(e) \subseteq \Gamma({\tt x})$. Therefore, $\sigma \Gequiv{\tt x} \rho$
  implies that $\sigma(e) \mathop{=} \rho(e)$
  and $\sigma'({\tt x}) \mathop{=} \sigma(e) \mathop{=} \rho(e) \mathop{=} \rho'({\tt x})$.
  For all variables {\tt y} other than {\tt x}, the store is left unchanged and thus
  $\sigma({\tt y}) \mathop{=} \sigma'({\tt y})$ and therefore $\sigma({\tt y}) \mathop{=} \rho({\tt y})$
  implies $ \sigma'({\tt y}) \mathop{=} \rho'({\tt y})$.

  \item $< {\tt out}\ e\ {\tt on}\ a\ {\tt @}\ p, \sigma > \overarrow{(a,v,p)}{} < {\tt skip}, \sigma >$
  with $\sigma(e) = v$.

  As this step leaves the store unchanged, the property holds trivially.

  \item $< {\tt if}\ e\ c_1\ c_2, \sigma > \overarrow{\epsilon}{} < c_1, \sigma >$
  with $\sigma(e) \not= 0$.

  If $\pc \not\in \Gamma({\tt x})$
  Lemma~\ref{app:lemma:nopcid} gives that $\Gamma({\tt x}) = \{{\tt x}\}$.
  Lemma~\ref{app:lemma:pc} tells us that $\sigma({\tt x}) = \sigma''({\tt x})$
  and $\rho({\tt x}) = \rho''({\tt x})$
  thus by $\sigma({\tt x}) = \rho({\tt x})$ we are done.
  Otherwise, $\pc \in \Gamma({\tt x})$ implies by \ruleName{\TSIfElse} that
  $\mathit{fv}(e) \subseteq \Gamma({\tt x})$ which means that 
  $\sigma(e) = \rho(e)$ and the case follows by induction on the shorter
  execution trace of $c_1$.

  \item $< {\tt if}\ e\ c_1\ c_2, \sigma > \overarrow{\epsilon}{} < c_2, \sigma >$
  with $\sigma(e) = 0$.

  Analogous to $\sigma(e) \not= 0$.

  \item $< {\tt while}\ e\ c, \sigma > \overarrow{\epsilon}{} < 
      {\tt if}\ e\ (c; {\tt while}\ e\ c)\ ({\tt skip}), \sigma >$.

  By Lemma~\ref{app:lemma:whileunfold},
  $\Gamma'({\tt x}) \subseteq \Gamma({\tt x})$
  and as $\sigma = \sigma'$ and $\rho = \rho'$ the case follows by induction.
  \end{itemize}
\end{proof}

\begin{lemma}[Channel Context Bounds]
  \label{app:lemma:contexts}
  The typing $|- \{c\} \Gamma$ guarantees that:
  \begin{enumerate}
    \item \label{app:lemma:contexts:length} If $<c,\sigma> \overarrowA{t}{*} <{\tt skip},\sigma'>$
          and $\sigma \Gequiv{a} \rho$
          then $<c,\rho> \overarrowA{t'}{*}$ has $\length{t} \geq \length{t'}$,
          and if also $<c,\rho> \overarrowA{t'}{*} <{\tt skip},\rho'>$
          then $\length{t} = \length{t'}$.
    \item \label{app:lemma:contexts:points} If $a \not\in \Gamma(a_p)$
          and $<c,\sigma> \overarrowA{t}{*}$
          then $a_p \not\in t$.
  \end{enumerate}
\end{lemma}

\begin{proof}
  For property~\ref{app:lemma:contexts:points} the reasoning is anologous to that for
  property~\ref{app:lemma:pc:points} in Lemma~\ref{app:lemma:pc}. The rest of the proof is for
  property~\ref{app:lemma:contexts:length}.

  By complete induction on $n$ in $<c,\sigma> \overarrowA{t}{n} <{\tt skip},\sigma'>$.
  If $n=0$ the property holds trivially.
  We then have $<c,\sigma> \overarrow{\alpha}{} <c',\sigma'> \overarrowA{t}{n}$ $<{\tt skip},\sigma''>$
  and $<c,\rho> \overarrow{\beta}{} <c'',\rho'> \overarrowA{t'}{*}$.

  By cases on the first reduction step:

  \begin{itemize}

  \item $< {\tt skip}; c, \sigma > \overarrow{\epsilon}{} < c, \sigma >$.

  $\Gamma = \Gamma';\Gammaid$ so $\sigma \Gequiv{a} \rho$
  implies $\sigma \NGequiv{\Gamma'}{a} \rho$ and as 
  $\sigma = \sigma'$, $\rho = \rho'$ and $c' = c''$ the case follows by induction.

  \item  $< c_1; c_2, \sigma > \overarrow{\alpha}{} < c_1'; c_2, \sigma' >$
  with $< c_1, \sigma > \overarrow{\alpha}{} < c_1', \sigma' >$.

  Let $|-\{c_1\}\Gamma_1$. By property~\ref{app:lemma:presdep:channels} of Lemma~\ref{app:lemma:presdep},
  $\Gamma_1(a) \subseteq \Gamma(a)$, thus $\sigma \Gequiv{a} \rho$ implies $\sigma \NGequiv{\Gamma_1}{a} \rho$.
  As the command fully evaluates, so does $c_1$: $<c_1,\sigma> \overarrowA{t_1}{n_1} <{\tt skip},\sigma''>$
  and therefore also $<c_1;c_2,\sigma> \overarrowA{t_1}{n_1} <c_2, \sigma''>$.
  As $n_1 \leq n$, by induction we have that $<c_1,\rho> \overarrowA{t'_1}{*}$
  with $\length{t'_1} \leq \length{t_1}$ and if $<c_1,\rho> \overarrowA{t'_1}{*} <{\tt skip},\rho''>$
  then $\length{t'_1} = \length{t_1}$.

  As $\Gamma = \Gamma_2;\Gamma_1$, for all ${\tt x} \in \Gamma_2(a)$ we have
  $\sigma \NGequiv{\Gamma_1}{\tt x} \rho$, and by Lemma~\ref{app:lemma:vars}
  $\sigma''({\tt x}) = \rho''({\tt x})$ which means that
  $\sigma'' \NGequiv{\Gamma_2}{a} \rho''$ and
  the property follows by induction on the execution of length $n - n_1$.

  \item $< {\tt x} := e, \sigma > \overarrow{\epsilon}{} < {\tt skip}, \sigma[{\tt x} \mapsto v] >$
  with $\sigma(e) = v$.

  As this configuration never produces any output, this property holds trivially.

  \item $< {\tt out}\ e\ {\tt on}\ a\ {\tt @}\ p, \sigma > \overarrow{(a,v,p)}{} < {\tt skip}, \sigma >$
  with $\sigma(e) = v$.

  As this configuration produces only one output, always on the same channel,
  this property holds trivally.

  \item $< {\tt if}\ e\ c_1\ c_2, \sigma > \overarrow{\epsilon}{} < c_1, \sigma >$
  with $\sigma(e) \not= 0$.

  If $\pc \not\in \Gamma(a)$ this means that $c$ never produces an output on $a$
  by property~\ref{app:lemma:pc:channels} of Lemma~\ref{app:lemma:pc} and we are done.

  Otherwise, $\pc \in \Gamma(a)$ and therefore ${\it fv}(e) \subseteq \Gamma(a)$
  and $\rho(e) = \sigma(e) \not= 0$.
  Thus also $<c,\rho> \overarrow{\epsilon}{} <c_1,\rho>$.
  As $\Gamma_1(a) \subseteq \Gamma(a)$, this follows directly by induction.

  \item $< {\tt if}\ e\ c_1\ c_2, \sigma > \overarrow{\epsilon}{} < c_2, \sigma >$
  with $\sigma(e) = 0$.

  Analogous to when $\sigma(e) \not= 0$.

  \item $< {\tt while}\ e\ c, \sigma > \overarrow{\epsilon}{} < 
      {\tt if}\ e\  (c; {\tt while}\ e\ c)\ {\tt skip}, \sigma >$.

  By Lemma~\ref{app:lemma:whileunfold}, $\Gamma'(a) \subseteq \Gamma(a)$
  and as $\sigma = \sigma'$ and $\rho = \rho'$ the case follows by induction.

  \end{itemize}
\end{proof}

\begin{lemma}[Preserving Dependencies]
  \label{app:lemma:presdep}
  The dependencies of program points and context bounds only
  increase with sequential composition:
  \begin{enumerate}
  \item \label{app:lemma:presdep:points}
  If $|- \{ c_1 \} \Gamma_1$ and $|- \{c_1;c_2\} \Gamma$ then $\Gamma_1(a_p) \subseteq \Gamma(a_p)$
  for all $a_p$.
  \item \label{app:lemma:presdep:channels}
  If $|- \{ c_1 \} \Gamma_1$ and $|- \{c_1;c_2\} \Gamma$ then $\Gamma_1(a) \subseteq \Gamma(a)$
  for all $a$.
  \end{enumerate}
\end{lemma}

\begin{proof}
  For property \ref{app:lemma:presdep:points}.
  First, we establish that for all typings $|- \{c\} \Gamma$, $a_p \in \Gamma(a_p)$
  by induction on the typing.
  This property clearly holds for \ruleName{\TSSkip}, \ruleName{\TSAssign}
  and \ruleName{\TSOutput}.
  For \ruleName{\TSSequence}, by induction $a_p \in \Gamma_2(a_p)$.
  Thus $\Gamma_1(a_p) \subseteq \Gamma(a_p)$. By induction we also have
  $a_p \in \Gamma_1(a_p)$ and thus $a_p \in \Gamma(a_p)$.
  For \ruleName{\TSIfElse}, we have $\Gamma_1(a_p) \subseteq \Gamma(a_p)$ and it
  follows by induction.
  For \ruleName{\TSWhile}, we have $\Gamma_c(a_p) \subseteq \Gamma(a_p)$ since 
  $\Gamma_c;\Gammaid[\pc \mapsto \{\pc\} \cup {\it fv}(e)]$ is in the union that
  defines $\Gamma$ and this case follows by induction as well.

  So if $|-\{c_1\}\Gamma_1$ and $|-\{c_2\}\Gamma_2$ this means
  $|-\{c_1;c_2\}\Gamma_2;\Gamma_1$. As $a_p \in \Gamma_2(a_p)$,
  we have $\Gamma_1(a_p) \subseteq \Gamma(a_p)$ which is what we had to show.

  Similar for property \ref{app:lemma:presdep:channels}.
\end{proof}

\begin{lemma}[Non-Branching Reduction]
  \label{app:lemma:nbr}
  For all configurations $c$ and stores $\sigma, \rho$ such that
    $|- \{c\} \Gamma$ with
    $<c,\sigma> \overarrow{\alpha}{} <c',\sigma'>$ and
    $<c,\rho> \overarrow{\beta}{} <c',\rho'>$,
    i.e.\ to the same command $c'$,
    where $\alpha, \beta$ is either an output or $\epsilon$. Then
    $\sigma \Gequiv{a_p} \rho$ implies
    that $|- \{c'\} \Gamma'$ and
    $\sigma' \NGequiv{\Gamma'}{a_p} \rho'$.
\end{lemma}

\begin{proof}
  By induction on the small step evalation.

  \begin{itemize}

  \item $<{\tt skip;} c,\sigma> \overarrow{\epsilon}{} <c,\sigma>$
    and $<{\tt skip;} c,\rho> \overarrow{\epsilon}{} <c,\rho>$
    
    Typing gives us $\Gamma = \Gamma_2; \Gammaid$ and $\Gamma' = \Gamma_2$. 
    As $\Gamma_2;  \Gammaid = \Gamma_2$ we have that $\Gamma \mathop{=} \Gamma'$,
    $\sigma \mathop{=} \sigma'$ and $\rho \mathop{=} \rho'$;
    and $\sigma' \NGequiv{\Gamma'}{a_p} \rho'$ holds trivially.

  \item $<c_1;c_2,\sigma> \overarrow{\alpha}{} <c_1';c_2,\sigma'>$
    with $<c_1,\sigma> \overarrow{\alpha}{} <c_1',\sigma'>$ 
    and\\ $<c_1;c_2,\rho> \overarrow{\beta}{} <c_1';c_2,\rho'>$
    with $<c_1,\rho> \overarrow{\beta}{} <c_1',\rho'>$
    
    By induction there exists a $\Gamma'_1$ such that
    $|- \{c'_1\} \Gamma'_1$ and $\sigma' \NGequiv{\Gamma'}{a_p} \rho'$.
    $\Gamma = \Gamma_2 ; \Gamma_1$ and
    $\Gamma' = \Gamma_2 ; \Gamma'_1$.   
    When $\sigma \NGequiv{\Gamma_2;\Gamma_1}{a_p} \rho$ then by definition of $;$, 
    $\forall {\tt y} \in \Gamma_2(a_p) . \sigma \NGequiv{\Gamma_1}{\tt y} \rho$. 
    By induction we also have that $\forall {\tt y} \in \Gamma_2(a_p) . \sigma' \NGequiv{\Gamma'_1}{\tt y} \rho'$.
    Thus $\sigma' \NGequiv{\Gamma_2;\Gamma'_1}{a_p)} \rho'$.

  \item $<{\tt x :=\,}e,\sigma> \overarrow{\epsilon}{} <{\tt skip}, \sigma[{\tt x} \mapsto v]>$
    with $\sigma(e) = v$ and\\
    $<{\tt x :=\,}e,\rho> \overarrow{\epsilon}{} <{\tt skip}, \rho[{\tt x} \mapsto w]>$
    with $\rho(e) = w$
    
    As $\Gamma' = \Gammaid$, $\sigma' \NGequiv{\Gammaid}{a_p} \rho'$ holds for any
    $\sigma', \rho'$.

  \item $<{\tt out}\, e\, {\tt on}\, a\, {\tt @ }\, p, \sigma> \overarrow{(a,v,p)}{}
    <{\tt skip},\sigma>$ with $\sigma(e) = v$ and\\ $<{\tt out}\, e\, {\tt on}\, a\, {\tt @ }\, p, \rho> \overarrow{(a,w,p)}{}
    <{\tt skip},\rho>$ with $\rho(e) = w$
    
    Similar as for assignment.

  \item $<{\tt if}\ e\ c_1\ c_2,\sigma> \overarrow{\epsilon}{} <c_1,\sigma>$
    with $\sigma(e) \not= 0$ and\\ $<{\tt if}\ e\ c_1\  c_2,\rho> \overarrow{\epsilon}{} <c_1,\rho>$

    As $\sigma' = \sigma$ and $\rho' = \rho$ we need to show that $\Gamma_1(a_p) \subseteq \Gamma(a_p)$
    which is clear from the typing rule.

  \item $<{\tt if}\ e\ c_1\ c_2,\sigma> \overarrow{\epsilon}{} <c_2,\sigma>$
    with $\sigma(e) = 0$ and\\ $<{\tt if}\ e\ c_1\  c_2,\rho> \overarrow{\epsilon}{} <c_2,\rho>$

    As $\sigma' = \sigma$ and $\rho' = \rho$ we need to show that $\Gamma_2(a_p) \subseteq \Gamma(a_p)$
    which is clear from the typing rule.
    
  \item $<{\tt while}\,e\,c,\sigma> \overarrow{\epsilon}{} <{\tt if}\ e\ (c; {\tt while}\,e\,c)\ {\tt skip},\sigma>$
    and\\ $<{\tt while}\,e\,c,\rho> \overarrow{\epsilon}{} <{\tt if}\ e\ (c; {\tt while}\,e\,c)\ {\tt skip},\rho>$
    
    As $\sigma' = \sigma$ and $\rho' = \rho$ it follows from $\Gamma'(a_p) \subseteq \Gamma(a_p)$
    by Lemma~\ref{app:lemma:whileunfold}.

  \end{itemize}
\end{proof}

\begin{lemma}[Branching Reduction]
\label{app:lemma:br}
For all configurations $c$ and stores $\sigma, \rho$ such that
  $|- \{c\} \Gamma$ with
  $<c,\sigma> \overarrow{\epsilon}{} <c'_\sigma,\sigma'>$ and
  $<c,\rho> \overarrow{\epsilon}{} <c'_\rho,\rho'>$,
  where $c'_\sigma \mathop{\not=} c'_\rho$. 
  For all $a_p$ such that
  $<c'_\sigma,\sigma'> \overarrowA{t(a,v,p)}{*}$, it holds that 
  if $\sigma \Gequiv{a_p} \rho$ then either:
\begin{itemize}
  \item There exists a joining command $c_j$ with $t_1 \mathop{\preceq} t$,
        such that
        $<c'_\sigma,\sigma'> \overarrowA{t_1}{*} <c_j,\sigma_j>$
        and
        $<c'_\rho,\rho'> \overarrowA{t'_1}{*} <c_j,\rho_j>$
        with $\length{t_1} \mathop{=} \length{t'_1}$ and
        $|- \{c_j\} \Gamma_j$ with $\sigma_j \NGequiv{\Gamma_j}{a_p} \rho_j$.
        I.e. both executions join again at equal command $c_j$ after an equal number
        of outputs with equivalent dependencies for $a_p$, or
  \item For all $t'$ such that 
        $<c,\rho> \overarrowA{t'}{*}$, $\length{t} \geq \length{t'}$.
\end{itemize}
\end{lemma}

\begin{proof}
  By induction on the small step evalation $<c,\sigma> \overarrow{\epsilon}{} <c'_\sigma,\sigma'>$.
  There are only three evaluations that may result in a different command for
  different stores, the others hold trivially.

  \begin{itemize}

    \item $<{\tt if}\ e\ c_1\ c_2,\sigma> \overarrow{\epsilon}{} <c_1,\sigma>$
      with $\sigma(e) \not= 0$

      \begin{itemize}
        \item If ${\tt pc} \in \Gamma(a_p)$ then the typing rules give that $\sigma(e) = \rho(e)$
          and both executions take the same branch;
          hence $c'_\sigma = c'_\rho$ so we contradict our assumption.
        \item If ${\tt pc} \not\in \Gamma(a_p)$ then by Lemma~\ref{app:lemma:pc} $<c,\sigma>$
          should not produce an output on $a_p$, which contradicts our assumption.
      \end{itemize}

    \item $<{\tt if}\ e\ c_1\ c_2,\sigma> \overarrow{\epsilon}{} <c_2,\sigma>$
      with $\sigma(e) = 0$
      
      Similar as for $\sigma(e) \not= 0$

    \item $<c_1;c_2,\sigma> \overarrow{\alpha}{} <c_1';c_2,\sigma'>$
      with $<c_1,\sigma> \overarrow{\alpha}{} <c_1',\sigma'>$ 
      
      \begin{itemize}
        \item $c_1$ produced $(a,v,p)$, that is $<c_1,\sigma> \overarrowA{t\cdot(a,v,p) }{*}$,
          By property~\ref{app:lemma:presdep:points} of Lemma~\ref{app:lemma:presdep}; $\Gamma_1(a_p) \subseteq \Gamma(a_p)$, therefore 
          $\sigma \Gequiv{a_p} \rho$ implies $\sigma \NGequiv{\Gamma_1}{a_p} \rho$
          and the case follows by induction.
          
        \item $c_2$ produces $(a,v,p)$.
              That is, $c_1$ evaluates completely $<c_1,\sigma> \overarrowA{t_1}{*} <{\tt skip},\sigma_j>$
              and thus $<c_1;c_2,\sigma> \overarrowA{t_1}{*} <c_2,\sigma_j>$ $\overarrowA{t_2\cdot(a,v,p)}{*} $. 
              By Lemma~\ref{app:lemma:contexts}, property \ref{app:lemma:contexts:points}; 
              as $c_2$ produces an output on $a_p$,
              it must be that $a \in \Gamma_2(a_p)$.
              As $\Gamma = \Gamma_2;\Gamma_1$ with $|- \{c_1\} \Gamma_1$, $|- \{c_2\} \Gamma_2$
              and $\sigma \Gequiv{a_p} \rho$
              this means that $\sigma \NGequiv{\Gamma_1}{a} \rho$.
              By Lemma~\ref{app:lemma:contexts}, property \ref{app:lemma:contexts:length} we get that
              for all $t_1'$ such that        
              $<c_1,\rho> \overarrowA{t_1'}{*}$ has $\length{t_1} \geq \length{t_1'}$
              and if $<c_1,\rho> \overarrowA{t_1'}{*} <{\tt skip},\rho_j>$
              then $\length{t_1} = \length{t_1'}$.
              If $c_1$ diverges this means that the second case of Lemma~\ref{app:lemma:br}
              holds and we are done.
              Otherwise we also have that 
              $<c_1,\rho> \overarrowA{t'_1}{*} <{\tt skip},\rho_j>$ and thus
              $<c_1;c_2,\rho> \overarrowA{t'_1}{*} <c_2, \rho_j>$.
              Again by $\Gamma = \Gamma_2;\Gamma_1$ and $\sigma \Gequiv{a_p} \rho$
              we have that $\forall {\tt x} \in \Gamma_2(l,p) . \sigma \NGequiv{\Gamma_1}{\tt x} \rho$.
              By Lemma~\ref{app:lemma:vars} it follows that $\sigma_j \NGequiv{\Gamma_2}{a_p} \rho_j$,
              which is what we had the show (with $c_j = c_2$).
              
      \end{itemize}
  \end{itemize}
\end{proof}

}

\end{document}